\documentclass[11pt,oneside]{article}
\usepackage{import}
\usepackage{etoolbox}

\usepackage[font=small,skip=0pt]{caption}

\providebool{DRAFT}
\booltrue{DRAFT}

\usepackage{etoolbox}
\usepackage{xparse}

\makeatletter%
\@ifclassloaded{book}
{%
  \newcommand{\whenbook}[1]{%
    #1%
  }%
}
\makeatother%

\providecommand{\whenbook}[1]{}%

\providebool{bookdraft}%

\NewDocumentEnvironment{When}{m +b}{%
  \providebool{#1}%
  \ifbool{#1}{#2}{}%
}{}%

\NewDocumentEnvironment{Unless}{m +b}{%
  \providebool{#1}%
  \ifbool{#1}{}{#2}%
}{}%

\providebool{DRAFT}
\ifbool{DRAFT}{
  \newcommand{\whendraft}[1]{#1}%
    \newcommand{\unlessdraft}[1]{}%
  }

\NewDocumentEnvironment{Draft}{+b}{%
  \whendraft{#1}%
}{%
}%

\NewDocumentEnvironment{DRAFT}{O{BlueViolet}}{%
  \begin{Draft}%
    \color{#1}%
  }{%
  \end{Draft}
}%

\usepackage[hyphens]{url}

\usepackage[style=alphabetic,natbib=true,maxalphanames=4,minalphanames=3,maxnames=10,doi=false,isbn=false,url=false]{biblatex} %

\AtEveryBibitem{%
\ifentrytype{proceedings}{
    \clearfield{editor}%
    \clearname{editor}%
  }{}
  \ifentrytype{inproceedings}{
      \clearfield{editor}%
      \clearname{editor}%
}{}
}

\usepackage{amsmath}                               %
\usepackage{amsthm}                                %
\usepackage[anchorcolor=blue,colorlinks]{hyperref}%
\usepackage{varioref}
\usepackage[capitalize]{cleveref}

\crefname{invariantsi}{invariant}{invariants}

\usepackage[margin=1in]{geometry}

\usepackage{fancyhdr}

\usepackage[spacing,kerning=true]{microtype}          %
\usepackage[T1]{fontenc}

\usepackage{article} %
\usepackage{math} %
\usepackage{algo} %
\usepackage{wrapfig} %
\usepackage{placeins} %
\usepackage{advdate}%

\usepackage{appendix}

\usepackage{graphicx} %
\usepackage{layout} %
\usepackage{setspace} %
\usepackage{mathtools} %
\usepackage{grffile} %
\usepackage{pdfpages} %
\usepackage{mdframed} %
\usepackage{bm} %
\usepackage{needspace} %
\usepackage{pbox} %
\usepackage{cancel}
\usepackage[normalem]{ulem}

\usepackage[full]{textcomp}     %
\usepackage{lmodern}            %
\usepackage[T1]{fontenc}        %
\usepackage{inconsolata}        %

\usepackage{titlesec}

\usepackage{truncate}

\NewDocumentCommand{\undernote}{s O{blue} m m}{
  \IfBooleanT{#1}{\smash}%
  {\color{#2} %
    \underbrace{\normalcolor%
      #4}_{\mathclap{\text{#3}}} %
  }%
  \IfBooleanT{#1}{\vphantom{#4}}
}

\newcommand{\defterm}[1]{{\boldmath\normalfont \bfseries #1}}%
\renewcommand{\defterm}{\emph}%

\makeatletter
\g@addto@macro\bfseries{\boldmath}
\makeatother

\usepackage{figs}

\titlespacing*{\paragraph}{%
  0pt}{
  {\medskipamount}}{
  1em}

\titleformat{\subparagraph}[runin]{\itshape}{0pt}{}{}%

\titlespacing*{\subparagraph}{%
  0pt}{
  {\medskipamount}}{
  1em}

\newcommand{\NameColorComment}[3]{%
  \whendraft%
  {\renewcommand\thefootnote{\textcolor{#2}{\arabic{footnote}}}%
    \footnote{\color{#2}#1: #3}%
  }%
}%

\newlength{\myalphabet}                %
\newlength{\mywidth}                   %
\newlength{\mymargin}                  %

\providecommand{\wrt}{with respect to\xspace}%
\usepackage{skak}%

\bibliography{references}

\providecommand{\poly}{\fparnew{\operatorname{poly}}}

\definecolor{calypso}{RGB}{50, 104, 145} %

\hypersetup{allcolors=calypso}  %

\definecolor{almostblack}{RGB}{18, 18, 18} %

\color{almostblack}%
\hypersetup{allcolors=almostblack}  %

\providebool{article}
\booltrue{article}

\begin{document}

\newcommand{\W}{\mathcal{W}}
\newcommand{\potential}{\fparnew{\varphi}} %
\providecommand{\pot}{\potential}       %
\providecommand{\pote}{\potential}
\providecommand{\potentialB}{\fparnew{\psi}}      %
\providecommand{\poteA}{\pote}                    %
\providecommand{\poteB}{\potentialB}              %
\providecommand{\potentialC}{\fparnew{\chi}}      %
\providecommand{\poteC}{\potentialC}              %
\providecommand{\len}{\fparnew{\ell}}       %
\newcommand{\potlen}{\fparnew{\ell_{\potential}}} %
\newcommand{\plen}{\fparnew{\ell_{\potential}}}   %
\newcommand{\p}{\fparnew{\varphi}}                %
\newcommand{\lenp}{\fparnew{\ell_{\p}}}

\newcommand{\through}{\fparnew{\operatorname{T}}} %
\newcommand{\sandwich}{\fparnew{\operatorname{S}}} %
\newcommand{\disp}{\fparnew{d_{\p}}}                 %
\newcommand{\dis}{\fparnew{d}}                     %
\newcommand{\hopd}[1]{\fparnew{d^{#1}}}           %
\newcommand{\hopdp}[2][\p]{\fparnew{d^{#2}_{#1}}}            %
\newcommand{\shopd}[1]{\fparnew{\hat{d}^{#1}}}               %
\newcommand{\shopdp}[2][\p]{\fparnew{\hat{d}_{#1}^{#2}}}
\newcommand{\phopd}[1]{\fparnew{\hat{d}^{#1}}}               %
\newcommand{\phopdp}[2][\p]{\fparnew{\hat{d}_{#1}^{#2}}}
\providecommand{\deg}{\fparnew{\operatorname{deg}}} %
\providecommand{\indeg}{\fparnew{\operatorname{deg}^-}} %
\providecommand{\outdeg}{\fparnew{\operatorname{deg}^+}} %
\providecommand{\hopindeg}[1]{\fparnew{\operatorname{deg}_{-}^{#1}}} %
\providecommand{\hopoutdeg}[1]{\fparnew{\operatorname{deg}_{-}^{#1}}} %
\providecommand{\heavyin}[1]{H^{-}_{#1}} %
\providecommand{\hin}{\heavyin} %
\providecommand{\heavyout}[1]{H^{+}_{#1}} %
\providecommand{\hout}{\heavyout}
\providecommand{\reachin}[1]{R^{-}_{#1}} %
\providecommand{\reachout}[1]{R^+_{#1}}  %
\providecommand{\apxheavy}{\smash{\tilde{H}}} %
\providecommand{\negV}{N}                     %
\providecommand{\negE}{E^-}
\providecommand{\numN}{k}
\renewcommand{\mod}{\operatorname{mod}} %
\providecommand{\distance}{\fparnew{d}} %
\providecommand{\Time}{\fparnew{T}}

\providecommand{\varh}{\eta}             %
\providecommand{\varvarh}{\zeta}
\providecommand{\varvarvarh}{\gamma}
\providecommand{\supersource}{s^{\star}} %
\NewDocumentCommand{\envelope}{m m m e{_}}{
  B^{#1}%
  \IfNoValueF{#4}{_{#4}}%
  \parof{#2,#3}
}
\NewDocumentCommand{\negativeenvelope}{m m m e{_}}{
  \bar{B}^{#1}%
  \IfNoValueF{#4}{_{#4}}%
  \parof{#2,#3}
}
\newcommand{\env}{\envelope}%
\newcommand{\nenv}{\negativeenvelope}%
\NewDocumentCommand{\subhopd}{m}{\fparnew{D^{#1}}}
\NewDocumentCommand{\subpopd}{m}{\fparnew{\hat{D}^{#1}}}
\renewcommand{\bar}{\overline}%
\providecommand{\Heads}{\bar{U}}
\providecommand{\outcut}{\fparnew{\partial^+}}
\providecommand{\ideald}{\fparnew{\delta^{\star}}} %

\providecommand{\mapstart}{\fparnew{\pi_0}} %
\providecommand{\mapend}{\fparnew{\pi_1}}   %


\newcommand{\kent}{\NameColorComment{Kent}{Cerulean}}
\newcommand{\navid}{\NameColorComment{Navid}{Mulberry}}

\newcommand{\email}[1]{\href{mailto:#1}{\texttt{#1}}}

\newcommand{\KentNote}{%
  \email{krq@purdue.edu}. Purdue University, West Lafayette,
  Indiana. Supported in part by NSF grant CCF-2129816. %
}
\newcommand{\NavidNote}{%
  \email{navidt@illinois.edu}. University of Illinois at
  Urbana-Champaign, Urbana, Illinois. %
}


\title{From Hop Reduction to Sparsification for Negative~Length~Shortest~Paths}
\author{Kent Quanrud\footnote{\KentNote} \and Navid Tajkhorshid\footnote{\NavidNote}}

\maketitle

\begin{abstract}
  The textbook algorithm for real-weighted single-source shortest
  paths takes $\bigO{m n}$ time on a graph with $m$ edges and $n$
  vertices. A recent breakthrough algorithm by \citet{Fineman24} takes
  $\apxO{m n^{8/9}}$ randomized time. The running time was
  subsequently improved to $\apxO{mn^{4/5}}$ \cite{HJQ25a} and then
  $\apxO{m n^{3/4} + m^{4/5} n}$ \cite{HJQ26}.

  We build on the algorithms of \cite{Fineman24,HJQ25a,HJQ26} to
  obtain faster strongly-polynomial randomized-time algorithms for
  negative-length shortest paths. An important new technique in this
  algorithm repurposes previous ``hop-reducers'' from
  \cite{Fineman24,HJQ26} into ``negative edge sparsifiers'', reducing
  the number of negative edges by essentially the same factor by which
  the ``hops'' were previously reduced.  A simple recursive algorithm
  based on sparsifying the layered hop reducers of \cite{Fineman24}
  already gives an
  $\apxO{m n^{\smash{\sqrt{3}}-1}} < \bigO{mn^{.7321}}$ randomized
  running time, improving \cite{HJQ26} uniformly.

  We also improve the construction of the bootstrapped hop reducers in
  \cite{HJQ26} by proposing new sparse shortcut graphs replacing the
  dense shortcut graphs in \cite{HJQ26}. Integrating all three of layered
  sparsification, recursion, and sparse bootstrapping into the
  algorithm of \cite{HJQ26} gives new upper bounds of
  \begin{math}
    \bigO{mn^{.7193}}
  \end{math}
  randomized time
  for $m \geq n^{1.03456}$ and
  \begin{math}
    \bigO{\parof{mn}^{.8620}}
  \end{math}
  randomized time for $m \leq n^{1.03456}$.

  Lastly, concurrent work by \cite{LLRZ25} obtained an
  $\apxO{n^{2.5}}$ randomized time algorithm for the same problem, and
  along the way improved the running time of the ``betweenness
  reduction'' step in Fineman's framework. Dropping in this subroutine
  as a black box improves the running time of the simple recursive
  sparsification algorithm to
  $\apxO{m n^{1/\sqrt{2}}} \leq \bigO{mn^{.70711}}$, and a slightly
  modified recursive sparsification algorithm runs in
  $\bigO{m n^{.69562}}$ randomized time for $m \geq n^{1.0274}$ and
  $\bigO{(mn)^{0.850}}$ for $m \leq n^{1.0274}$.
\end{abstract}

\section{Introduction}

Single-source shortest paths problem is a classical problem in graph
algorithms and a standard topic of basic algorithms courses. The input
consists of a directed graph $G = (V,E)$ with edge lengths
$\len: E \to \reals$ and a source vertex $s \in V$. The goal is to output
either the distance from $s$ to every other vertex, or a negative
cycle.

The textbook dynamic programming algorithm runs in $\bigO{mn}$ time,
where $m$ is the number of edges, and $n$ is the number of
vertices. It was discovered in the 1950s independently by
\citet{Shimbel55,Ford56,Bellman58,Moore59} and remained the fastest
algorithm until a recent breakthrough $\apxO{mn^{8/9}}$ randomized
time algorithm by \citet{Fineman24}. Fineman developed a new framework
based on identifying ``remote'' sets of negative edges and
neutralizing them efficiently by ``hop reducers'', which we discuss in
greater detail below.

There are faster algorithms for important special cases including
several notable recent developments. For nonnegative edge lengths, the
widely-taught Dijkstra's algorithm takes $\bigO{m + n \log n}$ time
\cite{Dijkstra59,FT87}, and an exciting new algorithm from
\cite{DMMSY25} takes $\bigO{m \log^{2/3} n}$ time.  Weakly polynomial
algorithms for integral edges weights have also recently attained
significant milestones.  \cite{BNW22} gave the first nearly linear
time algorithm (\wrt the bit complexity) for shortest paths.
Concurrently, \cite{CKLPPS25} gave the first almost-linear time
algorithm (again, \wrt the bit complexity) for minimum cost flow,
which generalizes shortest paths.

We are interested in the fully general setting of
\cite{Shimbel55,Ford56,Bellman58,Moore59,Fineman24} where the edge
weights are real-valued. Since Fineman's surprising result, there have
been two followup works building on Fineman's framework to further
improve the running time. \cite{HJQ25} first improved the running time
to $\apxO{m n^{4/5}}$ randomized time by introducing ``proper hop
distances'' to more efficiently extract remote sets. \cite{HJQ26}
improved the running time further to $\apxO{mn^{3/4} + m^{4/5} n}$
randomized time by developing a more efficient ``bootstrapped'' hop
reducer replacing Fineman's layered hop reducer. This paper takes the
next step in this line of research, with the following improved
bounds.

\begin{restatable}{theorem}{MainTheorem}
  \labeltheorem{best}\labeltheorem{sssp} Single-source shortest paths
  with real-valued edge lengths can be computed with high probability
  in $\apxO{mn^{\frac{7 - \sqrt{17}}{4}}}$ randomized time for
  $m \geq m_0 $ and $\apxO{\parof{mn}^{\frac{66-2\sqrt{17}}{67}}}$
  randomized time for $m \leq m_0$, where
  $m_0 = \bigO{n^{\frac{33-7\sqrt{17}}{4}}}$.  (Here we have
  $\frac{33-7 \sqrt{17}}{4} \approx 1.03456$,
  $\frac{7-\sqrt{17}}{4} \leq 0.719224$, and
  $\frac{66-2\sqrt{17}}{67} \leq 0.861997$.)
\end{restatable}

Perhaps more importantly than the exact running time, this work
introduces a few conceptual ideas on top of
\cite{Fineman24,HJQ25a,HJQ26} and to our understanding of
real-weighted shortest paths more generally.  There are really four
main new techniques. The first idea is a randomized sparsification
technique that transmutes hop reduction directly into a proportional
decrease in the number of negative edges. The second idea is that of
recursion, particularly to the sparsified graphs that now have
substantially fewer negative edges.  These first two techniques give a
relatively simple, recursive algorithm that runs in
$\apxO{mn^{\sqrt{3}-1}}$ randomized time, improving on \cite{HJQ26}
without any of the bootstrapping machinery.

The remaining two techniques enhance the boostrapping hop reducer
construction from \cite{HJQ26}. First, we develop a sparse auxiliary
``shortcut'' graph that substitutes for a dense one in \cite{HJQ26},
and helps address the larger running time in sparse graphs in
\cite{HJQ26} (i.e., the $\apxO{m^{4/5} n}$ term).  The second
technique ``boosts'' the bootstrap contruction by initiating the
bootstrapping from a larger ``base'' subgraph, where the larger base
case is supplied by the aforementioned method of sparsified
recursion. Altogether we obtain the final bound of \reftheorem{best}.

\paragraph{Independent work and improved running times.} The results
above were prepared and submitted for review in early November
\cite{Version1}. Concurrently, independent work by \citet{LLRZ25}
obtained an $\apxO{n^{2.5}}$ randomized time algorithm for
real-weighted shortest paths. Among other ideas, \cite{LLRZ25}
developed a faster ``betweenness reduction'' subroutine leveraging
recursion. Substituting this subroutine as a black box improves the
simple recursive algorithm from $\apxO{m n^{\sqrt{3}-1}}$ to
$\apxO{m n^{1/\sqrt{2}}} \leq \bigO{mn^{.70711}}$ randomized time, and
the bootstrap algorithm to $\bigO{mn^{.7044}}$ above some density
threshold. A third algorithm integrating a subset of the techniques
from the bootstrapping algorithm into the recursive sparsification
algorithm obtained the following running times.
\begin{theorem}
  \labeltheorem{sssp-updated} Single-source shortest paths with
  real-valued edge lengths can be computed with high probability in
  $\bigO{mn^{.69562}}$ randomized time for $m \geq n^{1.0274}$, and
  $\bigO{\parof{mn}^{0.85}}$ for $m \leq n^{1.0274}$.
\end{theorem}
We describe the new subroutine from \cite{LLRZ25}, update the bounds
for the two existing algorithms, and describe the third algorithm
attaining \reftheorem{sssp-updated}, in
\cref{recursive-betweenness-reduction}.

\section{High-level overview and techniques}

\labelsection{overview}

As the state-of-the-art algorithms have grown more complex and
technical over time, it may help to give a high-level overview
explaining the motivation and intuition behind the new techniques.  We
start by informally introducing a few notions and definitions key to
the framework established by \cite{Fineman24}; these and additional
definitions are described more formally in
\refsection{preliminaries}. To start, a ``hop'' refers to a negative
edge, and an ``$h$-hop walk'' is a walk with at most $h$ hops. The
``$h$-hop distance'' is the minimum length among all $h$-hop walks
from $u$ to $v$, and denoted $\hopd{h}{u,v}$. A ``proper walk'' walk
is one where all the negative edges are distinct.  A set of vertex
potentials $\varphi: V \to \reals$ induces edge lengths
$\plen{u,v} = \pote{u} + \len{u,v} - \pote{v}$ that preserve the
shortest path structure. We let $\distance{u,v}_{\pote}$ denote
distances \wrt $\plen$. A potential is \emph{valid} if it introduces
no new negative arcs, i.e. $\len{e}_\pote \geq 0$ if $\len{e} \geq
0$. We say $\pote$ neutralizes a negative edge $(u, v)$ if
$\len{u, v}_\pote \geq 0$. \citet{Johnson77} observed that the
potentials $\pote{v} = \distance{V,v} = \min_{s \in V} \distance{s,v}$
are valid and neutralize the entire graph. In general, any potential
of the form $\pote{v} = \hopd{h}{V,v}$ is valid.

The goal is to iteratively compute valid potentials $\varphi$ that
neutralize negative arcs at a rate faster than $\apxO{m}$ time per
neutralized arc. Let $k$ denote the number of negative vertices; by
standard preprocessing, $k \leq n$. \cite{Fineman24} introduced the
notion of \emph{remote} sets of negative vertices that can be
neutralized very efficiently. For a parameter $h \in \naturalnumbers$,
a set of negative vertices $U$ is ``$h$-remote'' if they can reach at
most $n/h$ vertices and $m/h$ edges via negative-length $h$-hop
walks.
\cite{Fineman24} observed that given an $h$-remote set of vertices,
one can construct an ``$h$-hop reducer'' $H$ for the subgraph $G_U$
with $\bigO{m}$ edges. This means that $H$ contains the vertices of
$G$ as a subset of $V_H$, and for any two vertices $u,v \in G$, and
hop parameter $\varh \in \naturalnumbers$,
\begin{align*}
  \distance{u,v}_G \leq \hopd{\roundup{\varh/h}}{u,v}_H \leq \hopd{\varh}{u,v}_G
\end{align*}
In particular, one can compute Johnson's potentials for $G$, as a
single-source, $(\sizeof{U} / h)$-hop distance computation in $H$ -- a
factor $h$ speedup. Meanwhile, \cite{HJQ25} showed that an $h$-remote
set of size $\bigOmega{\sqrt{kh}}$ can be obtained in $\apxO{m h^2}$
time. The $\apxO{m n^{4/5}}$ running time in \cite{HJQ25} is obtained
by repeatedly extracting $h$-remote sets and neutralizing them via
this $h$-hop reducer from \cite{Fineman24}, for
$h = \apxTheta{k^{1/5}}$.

The hop reducer from \cite{Fineman24} is a simple and intuitive
layered graph with $h + 1$ layers. Let $G_h = (V_h,E_h)$ be the
subgraph induced by the $h$-hop negative reach of $G^+$. The first
layer $L_0$ is a copy of the nonnegative part of the graph, $G^+$. The
remaining $h$ layers $L_1,\dots,L_{h}$ are each a copy of
$G_h^+$. For $v\in V$ and $i \in \setof{0,\dots,h}$, we let $v_i$
denote the copy of $v$ in $i$th layer $L_i$. Each negative arc $(u,v)$
in $U$ is added as a forward arc $(u_i,v_{i+1 \mod h})$ between
consecutive layers $L_i$ and $L_{i+1}$ for each layer index $i$. We
also add length-0 ``self'' arcs $(u_i,u_{i+1 \mod h})$ for each
$u \in V_h$ and each index $i$. Finally, for every arc
$(x,y) \in \outcut{V_h}$ leaving the $h$-hop negative reach $V_h$, and
each layer $L_i$ for $i > 1$, we add an ``exit'' arc $(x_i,y_0)$ with the same
edge length.

\begin{center}
  \includegraphics{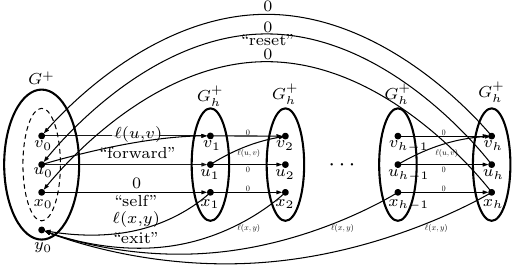}
\end{center}

Call the layered graph above $H'$. It is easy to see that $H'$
preserves distances in the sense that for $u,v \in V$,
$\distance{u,v}_{G_U} = \distance{u_0,v_0}_{H'}$. Now consider the
potentials $\varphi: V_{H'} \to \reals$ defined by
$\varphi(v_i) = \hopd{i}{V, u}$. In the reweighted graph
$H'' = H'_{\pote}$, the only negative edges are the arcs
$(u_{h+1},v_{0})$ wrapping around from the final layer back to the
first. Crucially, $\len{u_i, v_0}_\varphi \geq 0$ for every exit arc
$(u, v) \in \outcut{V_h}$ since $v$ is not in the $h$-hop negative reach
of $U$.  Now, for any $h$-hop walk $W: s \leadsto t$ in $G$, we have a
$1$-hop walk $W': s_0 \leadsto t_0$ in $H''$ where each hop in $w$ is
mapped to a forward arc from one layer to the next, before wrapping
back to get to $t_0$. It follows (with care) that $H''$ is an $h$-hop
reducer, preserving distances from $G$, while reducing the number of
hops in any walk by a factor of $h$.

\cite{HJQ26} proposes a more involved hop reducer construction based
on \emph{bootstrapping}, which we will describe in a moment.
Ultimately, the algorithms in \cite{Fineman24,HJQ25,HJQ26} all use
$h$-hop reducers to reduce the number of hops needed to compute
Johnson potentials for $G_U$.  The first technique we introduce,
recursive sparsification, is the first time an algorithm does
\emph{not} (directly) use hop reduction to speed up shortest path
computations.

Consider again the layered $h$-hop reducer $H''$ described above. All
the negative arcs ``reset'' from the last layer to the first.  Long,
many-hop walks through $V_h$ are mapped through the layers of $H''$, so
that $h$-hop subpaths become $1$-hop ``round trips'' using all the
layers of $H''$. Keeping the picture of a long walk $W$ through $G_h$ with
at least $h$ hops in mind, suppose we randomly sampled a subset
$U' \subseteq U$ where each negative vertex/arc is sampled
independently with probability $\bigOmega{\log{n} / h}$.  The random
sample $U'$ hits the walk once every $\bigOmega{h / \log{n}}$ hops on
average, and once every $h$ hops with high probability.

\begin{center}
  \includegraphics{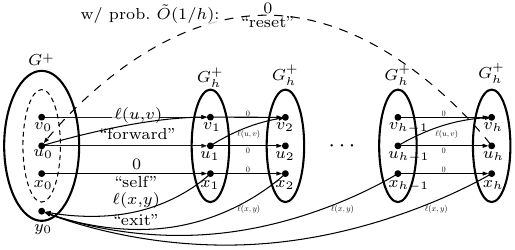}
\end{center}

Assume this high probability event.  Let $H$ be the subgraph of $H''$
obtained by restricting the negative reset arcs to those in $U'$. We
can route $W$ through the layers of the $H$, taking at most $h$
non-sampled hops through the layers before taking one of the sampled
``reset'' arcs. This observation, that walks through the negative
reach can still be mapped through $H$ with high probability, leads to
the fact that $H$ (essentially) preserves all distances in $G_U$, with
the following salient property: in contrast to $G_U$, $H$ has only
$\bigO{\sizeof{U} \log{n} /h}$ negative edges!

Now, if we can neutralize $H$, then we can compute the Johnson's
potentials for $G_U$ in a single Dijkstra computation over the
neutralized $H$. How do we neutralize $H$?

Recurse! This simple algorithm leads to a running time recursion of the form
\begin{align*}
  \Time{m,k} = \apxO{\sqrt{k/h}\parof{mh^2 + \Time{m, \sqrt{k/h}}}},
\end{align*}
This recursion is solved by $\Time{m,k} = \apxO{m k^{\sqrt{3} -1}}$,
where $\sqrt{3}-1 \approx 0.732051$. It is surprising that such a
relatively simple algorithm already improves \cite{HJQ26} for all
graph densities. It is also curious to note that this algorithm never
explicitly applies hop reduction to accelerate the computation of
Johnson's potentials. Evidently, it is more beneficial to exchange the
hop reduction for a proportional decrease in the number of negative
arcs.

\subsection{A leaner bootstrapped hop reducer}
\labelsection{bootstrap-overview}

The improvement from \cite{HJQ25} to \cite{HJQ26} comes primarily from a stronger hop
reducer. Let us say that a set of negative edges is $(h_1,h_2)$-remote
if it can $h_1$-hop negatively reach at most $n / h_2$ vertices. Thus
$U$ is $h$-remote if it is $(h,h)$-remote.  \cite{HJQ26} observed that
in essentially the same $\apxO{m h^2}$ running time to previously
compute an $h$-remote set $U$ of size $\apxOmega{\sqrt{kh}}$, one can
compute a set of negative edges $U$ of size $\apxOmega{\sqrt{k}}$ that
is $(h_1,h_1/h^2)$-remote for all $h_1 \geq \bigOmega{\log n}$. In
particular, the $\apxO{1}$-hop negative reach of $U$ has $n/ h^2$
vertices, and grows smoothly until the $h^2$-hop negative reach
becomes the entire graph.

The goal is to construct an $h^2$-hop reducer for $G_U$ in
$\apxO{mh^2 + m \sizeof{U}/ h^2}$ time. \cite{HJQ26} achieves this by
the following \emph{bootstrapping construction}.  Assume for
simplicity that $G$ is sufficiently dense. Consider the subgraphs
$G_i$ of $G_U$ induced by the $2^{i}$-hop negative reach for
$i = 1,...,2 \log h$. The incremental goal is to build a
$\bigOmega{2^{i}}$-hop reducer $H_i$, with $\bigO{m 2^i / h^2}$ edges,
for each $G_i$; for $i = 2 \log h$ this gives the desired
$\bigO{m}$-edge $\bigOmega{h^2}$-hop reducer for $G_{2 \log h} =
G_U$. Now, it is easy to efficiently construct a 2-hop reducer from
$G_1$ because $G_1$, with $\bigO{m / h^2}$ edges, is so small.  The
challenge is trickier for larger $i$. For example, for $i > \log h$,
the $2^i$-layered hop reducer described above has
$\bigOmega{m 2^{2i}/ h^2} > m$ edges.

The key idea is to use the previous hop reducers $H_{1},\dots,H_{i-1}$
to construct the next hop reducer $H_i$. An initial inspiration might
be to try to use $H_{i-1}$ to compute distances in $G_i$ between vertices in $U$
much faster; then those distances can be used to ``shortcut'' the
layered graph construction. The ``shortcut hop reducer'' would only be
one additional layer, with auxiliary forward arcs between the tail of
each negative edge to the head of each negative arc (from $U$), with
arc lengths based on the distances computed via $H_{i-1}$. The problem
here is that it takes too long to compute hop distances from every
$u \in U$, even when it is only a constant number of hops in $H_i$.

The actual construction in \cite{HJQ26} is more subtle. For each
$j < i$, \cite{HJQ26} uses the hop reducer $H_j$ to compute ``distance
estimates'' $\delta_j(u,v)$ between every pair $u,v \in U$ that are at
least as good as any proper $\bigTheta{2^j}$ hop walk from $u$ to
$v$. In particular, it does not try to compete with walks with
significantly less than $2^j$ hops. This restriction allows for the
following speedup: rather than compute single source hop distances
from every $u \in U$ in $\apxO{m \sizeof{U} 2^{j} / h^2}$ time, we
can randomly sample $\apxO{\sizeof{U} \log{n}/2^j}$ vertices
$X_j \subseteq U$ uniformly at random, and compute the
$\bigO{2^{j}}$-hop distances in $G_j$ (via $H_j$) to and from every
$x \in X_j$, in $\apxO{m \sizeof{U} / h^2}$ time. For every
$s, t \in U$, and sampled $u \in X_j$, the computed distances from $s$
to $u$ and then to $t$ gives a valid distance from $s$ to $t$.  With
high probability, $X_j$ ``hits'' all the shortest walks with at least
$\bigOmega{2^{j}}$ hops. As such, the distance estimates via $X_j$
``capture'' all the proper $\bigTheta{2^j}$-hop distances in $G_j$.

\FloatBarrier

At this point we have a distance estimate $\delta_j(u,v)$ that is
bounded above by the length of any $\bigOmega{2^{j}}$-hop $(u,v)$-walk
in $G_j$, for all $j < i$. To construct the $\bigOmega{2^i}$-hop
reducer $H_i$ for $G_i$, \cite{HJQ26} starts with disjoint copies of
$G_j^+$ for each index $j \leq i$. We let $v_j$ denote the copy of $v$
in $G_j^+$, for $v \in V$ and $j \leq i$. For every negative arc
$(u,v)$, and each index $j$, we add the ``shortcut'' arc $(u_i,v_j)$
with an edge length based on $\delta_j(u,v)$. (Omitting details, some
adjustment is made to $\delta_j(u,v)$ when the distance estimate is
observed to be too good to be true.) There are also ``exit'' arcs
$(y_j,z_i)$ for each arc $(y,z)$ leaving the $2^j$-hop negative
reach. Lastly, there are also length-0 ``reset'' arcs $(u_{i-1},u_i)$ for
each $u \in U$.

\begin{figure}[htb]
  \centering
  \includegraphics{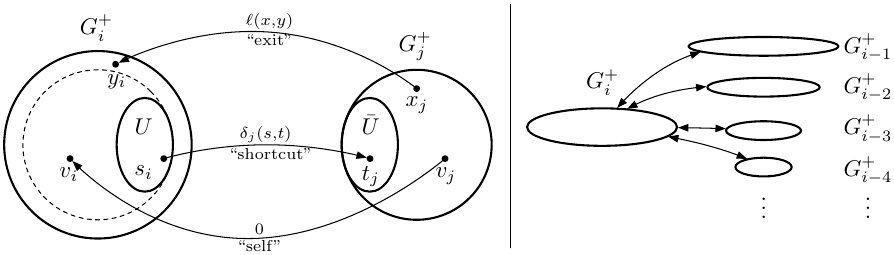}
\end{figure}

The overall hop reducer is arranged as a star, with a ``base'' copy of
$G_i^+$ in the center, connected to each $G_j^+$ with $j < i$ via the
gadget described above. Each $G_j^+$-gadget can be interpreted as a
compressed version of the layered graph construction for $G_j$, with
each ``roundtrip'' through $G_j^+$ reducing $2^j$ hops, with the
significant difference this gadget only ``captures'' walks with
$\bigTheta{2^j}$-hop walks, instead of any walk with $\bigO{2^j}$
hops. With an appropriate potential function, the only negative arcs
are the reset arcs from $G_{i-1}^+$ back to $G_i^+$. With a more
delicate argument than for the layered graph construction described
previously, one can argue that any walk in $G_i$ can be mapped through
our construction with $2^{i-1}$-factor less hops, giving us a
$2^{i-1}$-hop reducer $H_i$ for $G_i$. Bootstrapping in this fashion
until $i = 2 \log h$ leads to a desired $h^2$-hop reducer for $G$.

But consider the number of edges in the $i$th hop reducer $H_i$. The
disjoint copies of $G_j$ contribute at most $\bigO{m 2^i / h^2}$
edges, as desired. However, there are $\sizeof{U}^2 = \apxOmega{n}$
shortcut edges, which for small $i$ can be a bottleneck in
sufficiently sparse graphs. The net effect is that \cite{HJQ26}
obtains a $\apxO{m n^{3/4} + m^{4/5} n}$ running time; in particular,
$\apxO{m^{4/5} n}$ rather than $\apxO{m n^{3/4}}$ for
$m \leq \apxO{n^{5/4}}$.

\cite{HJQ26} is hamstrung by a mismatch between the sparsified
shortest path computation obtaining the distance estimates, and the
dense sets of shortcut edges based on the distance estimates. Recall
that the distance estimates $\delta_j$ for $G_j^+$ are based on
computing distances to and from a randomly sampled subset
$X_j \subseteq U$, while we add shortcut edges for every pair
$s,t \in U$. It is natural and more efficient to have a construction
that mirrors the computation, introducing an auxiliary copy of $X_j$
between $G_i^+$ and $G_j^+$, with edges from the copy of $U$ in
$G_i^+$ to the auxiliary copy of $X_j$ based on the distances to
$X_j$, and edges from the copy of $X_j$ to the copy of $U$ in $G_j^+$
based on the distances from $X_j$. Altogether we would have a
``bowtie'' connecting $G_i^+$ to $G_j^+$ via $X_j$, with
$2\sizeof{U} \sizeof{X_j}$ vertices.

\begin{center}
  \includegraphics{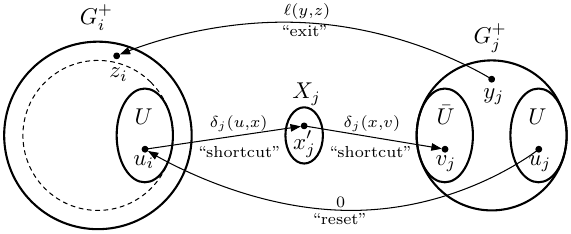}
\end{center}

Our next contribution is precisely such a construction.  Consider the
diagram above. For a vertex $u$ and index $j$, we let $u_j$ denote the
copy of $u$ in $G_j^+$ and $u_j'$ the copy of $u$ in $X_j$ (when these
auxiliary copies exist.)  For negative vertex $u \in U$ and
$x \in X_j$, we have a shortcut arc $(u_i, x_j')$ with a length
$\delta_j(u,x)$.  For each head of a negative edge, $v \in \Heads$,
and sampled vertex $x \in X_j$, we have another shortcut arc with a
have length $\delta_j(x,v)$. Observe that we only have
$\bigO{\sizeof{U} \sizeof{X_j}} = \apxO{\sizeof{U}^2 / 2^j}$ shortcut
edges. The subtle challenge comes in setting the right edge lengths
for $\delta_j(u,x)$ and $\delta_j(x,v)$.

The values $\delta_j(u,x)$ and $\delta_j(x,v)$ need to be ``sound'',
in the sense that they do not underestimate actual distances in
$G_j$. They should also be complete over all proper $\bigTheta(2^{h})$-hop walks
from $U$ to $\Heads$ in the following sense: for every pair $u \in U$ and
$v \in \Heads$, there is some $x \in X$ such that
$\delta_j(u,x) + \delta_j(x,v)$ is at most the length of any $(u,v)$
walk through $G_j$ with $\bigTheta{2^j}$ hops. The natural candidate
is to use the $2^j$-hop reducer from $H_j$ to compute $\bigO{2^j}$
hops starting and ending at each $x \in X_j$, and set $\delta_j(u,x)$
and $\delta_j(x,v)$ to these computed distances.

The second step of the hop reducer construction is to define valid
vertex potentials over $X_j$ and $G_j^+$ that neutralize all edges in
this gadget except for the reset arcs.  Our initial values for
$\delta_j(u,x)$ and $\delta_j(x,v)$ are not quite appropriate for this
task and require the following adjustments. First, for $u \in U$ and
$x \in X_j$, if the initially computed value $\delta_j(u,x)$ is less
than $\hopd{2^{j-1}}{V_j,x}_j$, we increase it to
$\hopd{2^{j-1}}{V_j,x}_j$.  Second, for $x \in X_j$ and
$v \in \Heads$, if the initially computed value of $\delta_j(x,v)$ is
less than $\hopd{2^j}{V_j,v} - \hopd{2^{j-1}}{V_j,x}$, we increase it
$\hopd{2^j}{V_j,v} - \hopd{2^{j-1}}{V_j,x}$. Increasing edge lengths
does not risk soundness, and one can argue that these increased edge
length are still complete, in the sense described above.

Now, with these adjusted edge lengths, we define a potential over
$X_j$ and $G_j^+$ by setting
\begin{math}
  \pote{x_j'} = \hopd{2^{j-1}}{V_j,x_j}_{j}
\end{math}
for $x \in X_j$
and
\begin{math}
  \pote{v_j} = \hopd{2^j}{V_j,v}_j
\end{math}
for $v \in V_j$, where $d_j(\cdot,\cdot)$ denotes distances in $G_j$. One
can show that these vertex potentials are valid over $G_j^+$ and the
exit arcs from $G_j^+$ to $G_i^+$, and also neutralize the shortcut arcs to
and from $X_j$. The only negative arcs remaining are the reset arcs.

Starting with $G_i^+$, and adding this ``sparse shortcut graph''
gadget for every $G_j^+$, $j < i$, gives a $\bigTheta{2^i}$-hop $H_i$
reducer for $G_i^+$. The $j$th shortcut graph now has
$\bigO{\sizeof{U}^2 / 2^j}$ shortcut edges, instead of
$\bigO{\sizeof{U}^2}$.

\paragraph{Boosting the bootstrapped hop reducer.}
Alas, the sparser shortcut-gadget construction does not reduce the
number of edges in the first shortcut gadget for $G_1^+$, and overall
the number of edges in the hop reducer has only decreased by a
logarithmic factor. The bottleneck is the density of the gadgets for
the smallest subgraphs $G_1^+, G_2^+, ...$.

Here we return to the ideas of layered sparsification and recursion to
boost the bootstrapping construction. The bootstrapping is
bottlenecked at the smallest negative-reach subgraphs. Instead, we
directly neutralize the subgraph induced by the $h_1$-hop negative
reach with layered sparsification and recursion, for a moderate value
of $h_1$, and start the bootstrapping process from there. Starting at
$h_1$-hop negative reach, instead of constant negative reach, also
allows us to increase the size of our remote set by a
$\poly{h_1}$-factor.

More precisely, let $i_0 \in \naturalnumbers$ be a parameter to be
determined, and let $i_1 = 2 i_0$. Let $h_0 = 2^{i_0}$ and
$h_1 = 2^{i_1} = h_0^2$. In $\apxO{m h}$ time, we sample a set of
$\apxOmega{\sqrt{k h_0}}$ negative vertices $U$ that is
$(\varh,h/\varh)$-remote for all $\varh \geq h_0$.  Let
$V_j$ be the $2^j$-hop negative reach, and
$G_j$ the subgraph of $G_U$ induced by $V_j$, as before.

We first apply the layered sparsification and recursive technique to
$G_{i_1}$, using $G_{i_0}$ as the layers. $G_{i_1}$ has $m h_1 / h$
edges, $n h_1/ h$ vertices, and $\sizeof{U}$ negative vertices.
$G_{i_0}$ has $m h_0 / h$ edges and $n h_0 / h$ vertices. By making
$h_0$ layers of $G_{i_0}$ over a base graph of $G_{i_1}$ and
sparsifying the negative edges, we obtain a graph $H_{i_1}'$ with
$\bigO{m h_1/h}$ edges, $\apxO{\sizeof{U} / h_0}$ negative edges,
and preserving distances from $G_{i_1}$. We recursively neutralize
$H_{i_1}'$, and use the neutralized $H_{i_1}'$ to compute Johnson
potentials and neutralize $G_{i_1}$.

In neutralizing $G_{i_1}$, we've also neutralized $G_j$ for all
$j < i_1$. These neutralized graphs can be used to compute the lengths
for the shortcut edges, and construct the sparse gadget described
above, for all $G_j$ with $i_0 < j \leq i_1$. For each such $j$, the
sparse gadget for $G_j$ captures all proper walks with
$\bigTheta{2^j}$ hops through $G_j$. Since we do not create any
gadgets for $j < i_0$, $G_{i_0}$ needs a different gadget to capture
all walks with at most $h_0$ hops in $G_{i_0}$. To this end we add the
familiar layered auxiliary graph, with $h_0$ copies of $G_{i_0}^+$.

\FloatBarrier

\begin{figure}[tbh]
  \centering
  \includegraphics{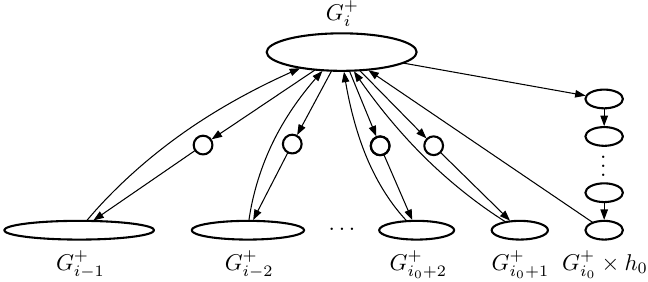}
\end{figure}

We now have a $\bigTheta{2^{i_1}}$-hop reducer $H_{i_1}$ for $G_{i_1}$, for a
remote set $U$ of $\apxOmega{\sqrt{k h_0}}$ negative vertices, with
the sparse gadgets starting from $G_{i_0 + 1}$, instead of
$G_1$. Continuing the bootstrap process from $H_{i_1}$ leads to an
$\bigTheta{h}$-hop reducer $H$ for all of $G_U$. Since the shortcut gadgets only
go down to $G_{i_0+1}$, each reducer has $\apxO{\sizeof{U}^2/h_0}$
added shortcut edges, instead of $\apxO{\sizeof{U}^2}$. Finally, we
compute Johnson's potentials neutralizing $G_U$ via $H$, in
$\apxO{m \sizeof{U}/h}$ time, as a $\sizeof{U}/h$-hop distance
computation. Choosing $i_0$ appropriately to balance parameters leads
to the $\bigO{mn^{.7193}}$ randomized running time for
$m \geq n^{1.03456}$. An additional preprocessing step (using layered
sparsification again!) leads to the $\bigO{(mn)^{.8620}}$ randomized
running time in sparser graphs.

Why didn't we sparsify and recurse on $H$, and on each intermediate
hop reducer $H_i$ along the way? Recursion wouldn't hurt, but it still
takes $\apxO{m \sizeof{U} / h}$ time to compute the distance estimates
$\delta_j(u,v)$ needed to label the shortcut edges in each sparse
gadget. Consequently the $\apxO{m \sizeof{U} / h}$ time spent directly
computing Johnson's potentials in $H$, instead of recursing, is not a
bottleneck.  Of course, one could have recursed on these hop reducers
anyway, resulting in an algorithm with the same running time that
never explicitly uses hop reduction to speed up a distance
computation.

In conclusion, this work introduces a new perspective recasting hop
reducers as negative-edge sparsifiers.  This step makes recursion more
viable and leads quickly to a relatively simple algorithm improving
the state of the art. These techniques were then incorporated into
bootstrapping (which this work also improves) to improve the running
time further.  Naturally, we don't believe these running times are the
best possible. At this stage, we are still developing basic techniques
and making elementary observations on a relatively new problem. We are
hopeful that the new ideas and techniques, especially the idea of
sparsification via hop reduction, will be helpful for future
algorithms.

\section{Preliminaries}

\labelsection{preliminaries}

Let $\mu = m + n \log n$ and let $\numN$ denote the number of negative
edges.

\paragraph{Distances.}
The distance from $s$ to $t$, denoted by $d(s,t)$, is defined as the
infimum length over all walks from $s$ to $t$.  For vertex
sets $S$ and $T$, we let $d(S,T)$ denote the minimum distance $d(s,t)$
over all $s \in S$ and $t \in T$.

\paragraph{Preprocessing and negative vertices.}

By standard preprocessing, we assume that the maximum in-degree and
out-degree are both $\bigO{m / n}$, that there are $\numN \leq n/2$
negative edges, and that each negative edge $(u,v)$ is the unique
incoming edge of its head $v$ and the unique outgoing edge of its tail
$u$ \cite{Fineman24,HJQ25}.  We identify each negative edge $(u,v)$
with its tail $u$, and call $u$ a \defterm{negative vertex}.  When
there is no risk of confusion, we may reference a negative edge by its
negative vertex and vice-versa. We let $N$ denote the set of negative
vertices. For a set of negative vertices $U \subseteq N$, we let
$\Heads = \setof{v \where u \in U \andcomma (u,v) \in E}$ be the set
of heads of the negative edges associated with $U$.

For a set of negative vertices $U$, we let $G_U$ denote the subgraph
obtained by restricting the set of negative edges to those
corresponding to $U$. We let $d_U(\cdot,\cdot) = d_{G_U}(\cdot,\cdot)$
denote distances in $G_U$.  We let $G^+ = G_{\emptyset}$ denote the
subgraph of nonnegative edges.

\paragraph{Hop distances and proper hop distances.}
The number of \emph{hops} in a walk is the number of negative edges
along the walk, counted with repetition. An \defterm{$h$-hop walk} is
a walk with at most $h$ negative edges.  The \defterm{$h$-hop
  distance} from $s$ to $t$ is the minimum length over all $h$-hop
$(s,t)$-walks.  Hop distances satisfy
\begin{align*}
  \hopd{h+1}{s,t} = \min{\hopd{h}{s,t},\, \min_{u \in N} \hopd{h}{s,u} +
  \len{u,v}+
  \hopd{0}{v,t}}. \labelthisequation{hop-distance}
\end{align*}
Given $\hopd{h}{s,u}$ for all $u$, one can compute
$\min_{u \in N} \hopd{h}{s,u} + \len{u,v} + \hopd{0}{v,t}$ for all $t$
with a single call to Dijkstra's algorithm over an appropriate
auxiliary graph. Thus \refequation{hop-distance} also describes a
dynamic programming algorithm computing single-source hop distances in
$\bigO{\mu}$ time per hop \cite{DI17,BNW22}.  With parent pointers,
one can also recover the underlying $h$-hop walks in the same running
time.

A \emph{proper $h$-hop walk} is defined as a walk with exactly $h$
negative edges and where all negative vertices are distinct.  For
vertices $s,t$, let
\begin{math}
  \shopd{h}{s,t}
\end{math}
denote the infimum length over all proper $h$-hop walks from $s$ to
$t$.  We call
\begin{math}
  \shopd{h}{s,t}
\end{math}
the \defterm{proper $h$-hop distance} from $s$ to $t$.

Proper hop distances are generally intractable.  For large $h$, proper
hop distances capture NP-Hard problems such as Hamiltonian path. They
were introduced in \cite{HJQ25} to extract larger ``negative
sandwiches'', a precursor to remote sets, with an argument that
bypasses this intractability.  In \cite{HJQ26}, proper hop distances
played a second role in the construction of the bootstrapped hop
reducer. Here, a random sampling technique used to compute the
distance estimates only works for walks with many distinct negative
edges.

\paragraph{Vertex potentials.}

Given vertex potentials $\p: V \to \reals$, let
$\lenp{e} \defeq \len{e} + \p{u} - \p{v}$ for an edge $e = (u,v)$. We
say that $\p$ \emph{neutralizes} a negative edge $e$ if
$\lenp{e} \geq 0$, and neutralizes a negative vertex if it neutralizes
the corresponding negative edge. \citet{Johnson77} observed that the
potential $\p{v} = d(V,v)$ neutralizes all edges (assuming there are
no negative-length cycles). We let $G_{\varphi}$ denote the graph with
edges reweighted by $\p$ and $\disp$ denote distances \wrt $\lenp$.

We say that potentials $\p$ are \emph{valid} if $\lenp{e} \geq 0$ for
all $e$ with $\len{e} \geq 0$; i.e., $\p$ does not introduce any new
negative edges. For any hop parameter $h \in \naturalnumbers$ and any
set of vertices $S$, the potentials $\p{v}_1 = \hopd{h}{S,v}$ and
$\p{v}_2 = -\hopd{h}{v, S}$ are valid.  Given two valid potentials
$\p_1$ and $\p_2$, $\max{\p_1,\p_2}$ and $\min{\p_1,\p_2}$ are also
valid potentials. If $\p_1$ is valid for $G$, and $\p_2$ is valid for
$G_{\p_1}$, then $\p_1 + \p_2$ is valid for $G$.

All algorithms discussed here follow an incremental reweighting
approach, dating back to \cite{Goldberg95}, where the graph is
repeatedly reweighted along valid potentials that neutralize a subset
of negative edges, until (almost) all edges are nonnegative. (A small
number of remaining negative edges can always be neutralized by
Johnson's technique in nearly linear time per negative edge.)  Then
the distances can be computed by Dijkstra's algorithm in the
reweighted graph.  Henceforth we focus exclusively on the effort to
neutralize all or almost all of the negative edges.

When algorithms reweight the graph along multiple potentials in one
iteration, we treat any initially negative edge as negative throughout
the iteration, even if incidentally neutralized. This is needed to
preserve hop-counting invariants. $k$ remains the number of negative
edges at the beginning of the iteration.

\paragraph{Remote edges.}

We say that a vertex $s$ can \defterm{$h$-hop negatively reach}
another vertex $t$ if $s = t$ or there is an $h$-hop $(s,t)$-walk with negative
length.  For parameter $h,r \in \naturalnumbers$, and a set of
negative vertices $U \subseteq V$, $U$ is \defterm{$(h,r)$-remote} if
the collective $h$-hop negative reach of $U$ has at most $n/r$
vertices. We say $U$ is \defterm{$h$-remote} if it is $(h,h)$-remote.

The algorithms here and in \cite{Fineman24,HJQ25,HJQ26} alternate
between extracting polynomially sized remote sets, and neutralizing
them more efficiently. The first step of extracting a remote set is
quite involved, and since the original method was introduced in
\cite{Fineman24}, has been improved in both subsequent works
\cite{HJQ25,HJQ26}.

This work focuses exclusively on the second step of neutralizing
remote sets, and uses previous techniques to extract the remote set in
a black box fashion. We will use the following lemma.

\begin{lemma}
  \labellemma{extract-sandwich}
  \labellemma{new-sparse-betweenness-reduction} Let
  $h_0 = \bigOmega{\log n}$ and $h \geq h_0$. One can compute, with
  high probability in
  $\bigO{h\mu \log^2 n + h^2\sqrt{k/h_0^3}\log^2 n}$ randomized time,
  either:
  \begin{compactmathresults}
  \item A negative cycle.
  \item Valid potentials $\varphi$ neutralizing a set of
    $\bigOmega{\sqrt{\numN h_0}}$ negative vertices.
  \item Valid potentials $\varphi$ and a set $U$ of negative vertices such that:
    \begin{compactmathproperties}
    \item $\sizeof{U} \geq \bigOmega{\sqrt{\numN h_0}}$.
    \item $U$ has $\varh$-hop negative reach of size $n\varh/h$ for
      all $\varh \geq h_0$.
    \end{compactmathproperties}
  \end{compactmathresults}
\end{lemma}

\reflemma{extract-sandwich} combines Lemmas 3.4 and 6.1 from
\cite{HJQ26}. The version in \cite{HJQ26} is stated for
$h_0 = \bigO{\log n}$; the same proof extends to a general parameter
$h_0$ as stated above.

\paragraph{Hop reducers.}

$h$-remote sets were introduced by \cite{Fineman24} to construct
linear-size hop reducers.  Formally, we say that a graph
$H = (V_H,E_H)$ is an \defterm{$h$-hop reducer} for a graph
$G = (V, E)$ if $V \subseteq V_H$ and
\begin{math}
  \distance{s,t}_{G} \leq \hopd{\roundup{\varh/h}}{s,t}_H \leq \hopd{\varh}{s,t}_{G}
\end{math}
for all $s,t \in V$ and $\varh \in \naturalnumbers$.

The layered graph from \cite{Fineman24}, described in
\refsection{overview}, was an $h$-hop reducer of size $\bigO{m}$. The
bootstrapped construction from \cite{HJQ26} gave an $h^2$-hop reducer
of size $\apxO{m}$.

\section{Layered sparsification and recursion}

\labelsection{layered-sparsification}

We first describe and analyze a relatively simple algorithm running in
$\apxO{m n^{\sqrt{3}-1}}$ randomized time.  This running time improves
the previous $\apxO{m n^{3/4} + m^{4/5} n}$ running time while cleanly
demonstrating two of the four main new techniques, layered
sparsification and recursion. Subsequent sections build on these
ideas to obtain even faster running times.

Layered sparsification recasts the layered auxiliary graph from
\cite{Fineman24} from a hop-reducer to a ``negative edge sparsifier'',
reducing the number of negative edges by essentially the same hop
reduction factor.
\begin{lemma}
  \labellemma{layered-sparsification} Let $U$ be a set of negative
  vertices with $h$-hop negative reach $n/r$. Let $m' = (2 + h/r)m + (h/r)n$ and
  $n' = (2 + h/r)n$. In $\bigO{m'}$ randomized time and with high
  probability, one can compute a graph $H$, potentials
  $\varphi: V_H \to \reals$, and mappings $\pi_0,\pi_1: V_G \to V_H$,
  such that:
  \begin{mathproperties}
  \item $H$ has $m'$ edges and $n'$ vertices.
  \item $H_{\varphi}$ has $\bigO{\sizeof{U} \log{n} / h}$ negative
    edges.
  \item \label{layered-sparsification-distances} For any two vertices $u,v \in V_G$,
    $d_{G_U}(u,v) = d_H(\pi_0(u),\pi_1(v))$.
  \end{mathproperties}
\end{lemma}

\begin{proof}
  For ease of notation we let $G = G_U$. We assume
  $h \geq \bigOmega{\log n}$, since otherwise we can take $H = G$.

  The auxiliary graph $H$ is an $(h+1)$-layer graph based on the
  layered $h$-hop reducer introduced in \cite{Fineman24}.  Let $G_h$
  be the $h$-hop negative reach of $U$.  We start with two copies of
  $G^+$ (layers $0$ and $h$) and $h-1$ copies of $G_h^+$ (layers $1$
  through $h-1$).  For a vertex $v$ and index
  $i \in \setof{0,\dots,h}$, let $v_i$ denote the copy of $v$ in layer
  $i$.  We connect these $h+1$ nonnegative subgraphs with four types
  of arcs.  First, for each layer $i \in \setof{0,\dots,h-1}$ and
  negative arc $(u,v)$, we add a ``forward arc'' $(u_i,v_{i+1})$ from
  the $i$th layer to the next with length $\len{u,v}$. Second, for
  each such layer $i$ and vertex $x \in G_h$ we also add a
  ``self-arc'' $(x_i,x_{i+1})$ with length $0$.  Third, for each index
  $i \in \setof{1,\dots,h}$ and each arc $(x,y) \in \outcut(G_h)$
  leaving $G_h$, we add an ``exit'' arc $(x_i,y_{0})$ of length
  $\len{x,y}$.  Lastly, let $U_0 \subseteq U$ sample
  $\bigO{\sizeof{U} \log{n} / h}$ negative vertices uniformly at
  random from $U$. For each sampled negative vertex $u \in U_0$, we
  add a ``reset arc'' $(u_h,u_0)$ with length $0$.  Lastly, we define
  the maps $\pi_0,\pi_1: V_G \to V_H$ by $\pi_0(v) = v_0$ and
  $\pi_1(v) = v_h$.

  \begin{center}
    \includegraphics{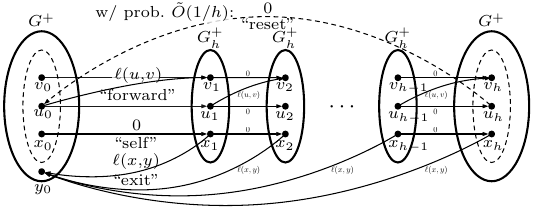}
  \end{center}

  We define potentials $\pote$ over $V_H$ by
  \begin{math}
    \pote{v_i} = d^i_{U}(V, v).
  \end{math}
  It is easy to see that $\pote$ is valid within each of the
  nonnegative subgraphs, and over the self-arcs. It is also
  straightforward to verify that $\pote$ neutralizes all forward
  arcs. As observed in \cite{Fineman24}, $\pote$ is valid over the
  exit arcs $(u_i,v_0)$, where $u \in G_h$ and $v \notin G_h$, precisely
  because $v$ is not in the $h$-hop negative reach of $U$. Thus the only
  negative arcs in $H_{\varphi}$ are the reset arcs $(u_h,u_0)$ for
  $U_0 \subseteq U$, of which there are
  $\bigO{\sizeof{U} \log{n} /h}$.

  It remains to prove \cref{layered-sparsification-distances}. The
  edges in $H$ are either length 0 self or reset arcs $(u_i,u_j)$
  between auxiliary copies of the same vertex $u$, or copies
  $(u_i,v_j)$ of an edge $(u,v)$ in $G$. By dropping the length-0 self
  and reset arcs, replacing auxiliary arcs $(u_i,v_j)$ with the
  underlying arcs $(u,v)$, every walk from $s_i$ to $t_j$ in $H$ maps
  to a walk from $s$ to $t$ in $G$ with the same distance. Thus
  $\distance{s_0,t_h}_H \geq \distance{s,t}_G$ for all
  $s,t \in V$.

  Towards proving the reverse inequality, we first claim
  that
  \begin{align}
    \distance{s_0,t_h}_H \leq \hopd{h}{s,t}_G \text{ for all } s,t \in
    V. \labelthisequation{embed-h-hops}
  \end{align}
  To this end, we prove that
  \begin{align*}
    \distance{s_0,t_h}_H \leq \phopd{\varh}{s,t}_G \text{ for all } s,t \in
    V \text{ and } \varh \leq h,
  \end{align*}
  by induction on $\eta$.

  In the base case, suppose we have a $0$-hop walk $W: s \leadsto
  t$. Let $W' : s_0 \leadsto t_h$ where we first map $W$ to the
  identical walk from $s_0$ to $t_0$ in the $0$th layer, and then take
  self arcs $(t_i,t_{i+1})$ from $t_0$ to $t_h$. $W'$ has the same
  length as $W$.

  In the general case, let $W: s \leadsto t$ be a proper $\varh$-hop
  walk of length $\phopd{\varh}{s,t}_G$ for $1 \leq \varh\leq h$.  We
  have two cases depending on whether $W$ stays in $G_h$ all throughout
  between its first and last hops.

  In the first case, suppose $W$ stays in $G_h$ between its first and
  last hops. Let $(x,y)$ be the last hop in $G_h$.  Let
  $W_1: s \leadsto x$ be the $(\varh-1)$-hop prefix of $W$ up to $x$
  and let $W_2 : y \leadsto t$ denote the suffix beginning at $y$.
  Consider the walk $W'_1: s_0 \leadsto x_{\varh - 1}$ in $H$, where we
  route $W_1$ in the natural way: starting from layer 0, we map the
  nonnegative arcs between the $i$th and $i+1$th hop to the
  corresponding nonnegative arcs in the $i$th layer, and we map the
  $i$th hop $(u,v)$ to the corresponding forward arc $(u_{i-1},v_{i})$
  from the $(i-1)$th layer to the $i$th. Let $W_2' : y_h \leadsto t_h$
  be the shortest $0$-hop walk from $y_h$ to $t_h$ in the $h$th layer,
  which we recall is a copy of $G^+$. To connect
  $W_1': s_0 \leadsto x_{\varh - 1}$ to $W_2' : y_h \leadsto t_h$, from
  $x_{\varvarh}$, we take the forward arc $(x_{\varh-1},y_{\varh})$, and
  then self-arcs $(y_i,y_{i+1})$ until reaching $y_h$. All together, we have an
  $(s_0, t_h)$-walk $W'$ whose length is identical to $W$.

  \begin{center}
    \includegraphics{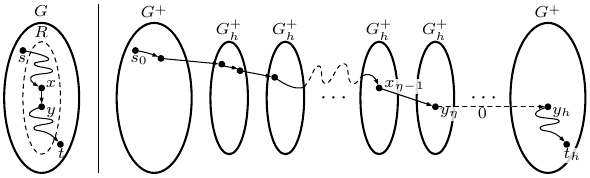}
  \end{center}

  In the second case, $W$ leaves $G_h$ between its first and last
  hop. Let $(x,y) \in \outcut{G_h}$ be the first arc where $W$ leaves
  $G_h$. Let $W_1: s \leadsto x$ be the prefix of $W$ up to $x$ and let
  $W_2 : y \leadsto t$ denote the suffix beginning at $y$.  Both $W_1$
  and $W_2$ have less than $\varh$ hops. By induction, there exist
  walks $W_1': s_0 \leadsto x_h$ and $W_2': y_0 \leadsto t_h$ no
  longer than $W_1$ and $W_2$, respectively. We connect $W_1'$ and
  $W_2'$ by inserting the exit arc $(x_h,y_0)$ between them. Together,
  $W_1'$, $(x_h,y_0)$, and $W_2'$ gives a $0$-hop $(s_0,t_h)$-walk no
  longer than $W$.

  \begin{center}
    \includegraphics{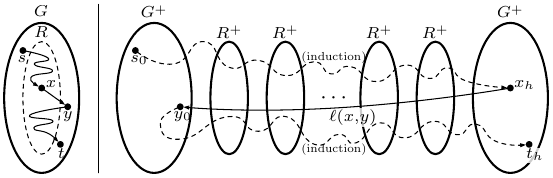}
  \end{center}

  This establishes inequality \refequation{embed-h-hops}.  Returning to
  \cref{layered-sparsification-distances}, fix any pair $s,t \in V$,
  and let $W: s \leadsto t$ be the shortest $(s,t)$-walk.  Of the many
  hops in $W$, some are from a vertex sampled in $U_0$ and some are
  not.  Let $u_1,u_2,\dots,u_{\ell} \in U_0$ be the sampled vertices
  visited by $W$, and divide $W$ into subwalks
  $W_1: s \to u_1, W_2: u_1 \to u_2,\dots,W_{\ell+1}: u_{\ell} \to t$
  at these points.

  Since $U_0$ samples each negative vertex with probability
  $\bigO{\log{n}/h}$, with high probability, there are no subsequences
  of $h$ consecutive hops in $W$ that do not include a vertex in
  $U_0$. In particular, each subwalk $W_i$ has at most $h$-hops.

  For each subwalk $W_i: u_i \to u_{i+1}$ (letting $u_0 = s$ and
  $u_{\ell+1} = t$), invoking the claim for $h$-hop distances above,
  let $W_i' : u_{i,0} \leadsto u_{i+1,h}$ be the walk in $H$ with
  length at most $\len{W_i}$. Let $W': s_0 \leadsto t_h$ be the walk
  obtained by connecting the subwalks $W_i' : u_i \to u_{i+1}$ with
  reset arcs $u_{i+1,h} \to u_{0,h}$ in between. We have
  \begin{math}
    \len_H(W') = \sum_i \len_H(W_i') \leq \sum_i \len{W_i} = \len{W},
  \end{math}
  as desired.
\end{proof}

The advantage of layered sparsification over hop reduction is that, by
reducing the number of negative edges, we have a materially smaller
subproblem amenable to recursion. Applying recursion in this
straightforward manner leads to the following.

\begin{theorem}
  The SSSP problem for real-weighted graphs can be solved with high
  probability in $\apxO{m n^{\sqrt{3} - 1}}$ randomized time.
\end{theorem}
\begin{proof}
  More precisely, we prove that a graph with $k$ negative edges and
  preprocessed as described in \refsection{preliminaries} can be
  neutralized in $\apxO{m k^{\sqrt{3}-1}}$ randomized time.

  The high-level idea combines \reflemma{layered-sparsification} with
  recursion. Let $h = \apxO{\smash{k^{\frac{2\sqrt{3} - 3}{3}}}}$. The
  algorithm iteratively neutralizes subsets of negative edges until
  there are no negative edges.  Each iteration, we invoke
  \reflemma{extract-sandwich} with parameters $h_0 = h$.
  With high probability, we either find and return a
  negative cycle, directly neutralize $\sqrt{k h}$ negative vertices
  and repeat, or obtain a valid potential $\varphi$ and a negative
  sandwich $U$ of $\bigOmega{\sqrt{k h}}$ negative vertices with
  $h$-hop negative reach $n / h$ in $G_{\varphi}$.

  In the latter scenario, we apply \reflemma{layered-sparsification}
  to $U$, obtaining the auxiliary graph $H$, potentials $\varphi$, and
  mappings $\pi_0,\pi_1$ as described in
  \reflemma{layered-sparsification} with high probability. $H$ has
  $\bigO{m}$ edges and $\bigO{n}$ vertices. We recurse on $H_{\pote}$
  to compute potentials $\poteB: V_H \to \reals$ neutralizing
  $H_{\pote}$. We then compute the Johnson potentials
  $\poteC{v} = d_U(V,v)$ for $G_U$ via the neutralized graph
  $H_{\pote,\psi}$, in $\bigO{m}$ time, neutralizing $U$.

  Each randomized step succeeds with high probability, and there are
  polynomially many steps, so the entire algorithm succeeds with high
  probability.

  Let $\Time{m,k}$ bound the running time on a graph with $m$ edges
  and $k$ negative edges. $\Time{m,k}$ is bounded recursively by
  \begin{align*}
    \Time{m,k} \leq \apxO{\sqrt{k / {h}} \parof{m h^2 + \Time{m,
    \sqrt{k / {h}}}}},
  \end{align*}
  which is satisfied by $\Time{m,k} = \apxO{m k^{\sqrt{3}-1}}$. For
  $k = n$ we obtain the desired running time.
\end{proof}

\section{Bootstrapping sparse hop reducers}
\labelsection{sparse-hop-reducers} \labelsection{sparse-bootstrap}

This section revisits and improves the bootstrapped hop reducer from
\cite{HJQ26}. The bootstrapping construction from \cite{HJQ26} was
described previously in \refsection{bootstrap-overview}. The
high-level idea, given a remote set of negative edges $U$, is to
gradually build out a hop reducer for $U$ working from the $\eta$-hop
negative reach for small $\eta$, and gradually extending $\eta$ until
we have a hop reducer for all of $G_U$. In principle, for the smallest
values $\eta$, the subgraph induced by the negative reach is small
enough for direct approaches, and then the hop reducer for one value
of $\eta$ can be used to help construct the hop reducer for the next
value of $\eta$. There are a number of details involved in this
high-level idea, and this section will make two important changes to
the bootstrapping construction described in \cite{HJQ26}.

Let $h, h_0 \in \naturalnumbers$ be parameters, to be determined, with
$h > h_0 \geq \bigOmega{\log n}$.  Let $U$ be a set of negative
vertices that is $(\varh, \varh/h)$-remote for all $\varh \geq h_0$.
The goal of this section is to construct an $h$-hop reducer for $G_U$.
Dropping any negative vertex outside of $U$, we assume that $G = G_U$.

The bootstrapping construction builds out hop reducers over the
$\varh$-hop negative reach of $U$ incrementally from small $\varh$ all
the way out to $h$. To index these steps, let $L = \logup[2]{h} + 1$,
$i_0 = \logup[2]{h_0}$, and $i_1 = 2i_0$. For each
$i \in \setof{i_0,\dots,L-1}$, let $V_i$ be the $2^i$-hop negative
reach of $U$, and let $G_i$ be the subgraph induced by $V_i$. $G_i$
contains $\bigO{2^in / h}$ vertices and $\bigO{2^i m / h}$ edges. Let
$V_L = V$ and $G_L = G$. For each $i$, let $d_i(\cdot, \cdot)$ denote
distances in $G_i$. Let $\outcut{V_i}$ be the directed cut of arcs
from $V_i$ to $V \setminus V_i$.

Following the bootstrapping framework of \cite{HJQ26}, for each index
$i$ successively, we create a $2^{i - 2}$-hop reducer $H_i$ over the
subgraph $G_i$. We construct the reducers incrementally, starting from
$H_{i_1 + 1}$, and use distances computed in previous hop reducers to
build the next hop reducer.  An important calculation motivating the
design is that each $H_i$ will have $\bigO{m 2^i/ h}$ edges, and will
be used for $\apxO{\sizeof{U}/2^i}$ shortest path computations (as
explained in \refsection{distance-estimates}). These two factors
cancel out nicely so that we end up spending $\apxO{m \sizeof{U} / h}$
time in each hop reducer $H_i$, independent of $i$.

Within this framework we make several changes. First, we do not start
the bootstrapping from the $\bigO{1}$-hop negative reach as
\cite{HJQ26}. The first hop reducer we construct will be for the
larger subgraph $G_{i_1}$. Ultimately, starting from a larger negative
reach will allow us to neutralize larger remote sets $U$ in each
iteration, as explained later in \refsection{altogether}.  To account
for the missing hop-reducers over the smaller negative-reach
subgraphs, we will assume we are given as an additional input
``distance estimates'' (defined momentarily) for all subgraphs from
$G_{i_0 + 1}$ to $G_{i_1}$, that suffice to initialize the bootstrap
process. Second, we create hop reducers for each subgraph $G_i$ that
are substantially sparser than those of \cite{HJQ26}, introducing
$\apxO{\sizeof{U}^2 / 2^i}$ additional edges for graph $G_i$ instead
of $\bigO{\sizeof{U}^2}$ additional edges as in \cite{HJQ26}. This
construction interacts synergistically with the decision to start the
bootstrapping at $G_{i_1}$ instead of $G_1$, as the $2^i$-factor is
much larger. The net effect is to improve the running time in sparse
graphs.

As alluded to above, the bootstrap construction uses shortest path
distances computed in the hop reducers for smaller subgraphs to
produce the next hop reducer for a larger subgraph. The distance
estimates computed via the smaller reducers will satisfy certain conditions
that are sufficient to construct the next hop reducer, which we formalize
via the notion of ``sparse distance estimates'' defined below. Sparse
distance estimates replace the (dense) distance estimates from
\cite{HJQ26}, and for a particular level $i$, will collectively be a
$\apxO{1/2^i}$-factor smaller than the set of dense distance
estimates at level $i$ in \cite{HJQ26}.

\begin{definition}
  We define \defterm{sparse distance estimates} at level $j$ as a set
  of negative vertices $X_j \subseteq U$, values $\delta_j(u,x)$ and
  $\delta_j(x,v)$ for $u \in U$, $x \in X_j$, and $v \in \Heads$, such
  that:
  \begin{mathproperties}
  \item \label{sparse-distance-estimates-1} $\delta_j(u,x) \geq \max{d(u,x), \hopd{2^{j-1}}{V_j,x}_j}$ for
    $u \in U$ and $x \in X_j$.
  \item \label{sparse-distance-estimates-2}
    $\delta_j(x,v) \geq \max{d(x,v), \hopd{2^{j}}{V_j,v}_j -
      \hopd{2^{j-1}}{V_j,x}_j}$ for all $x \in X_j$ and $v \in \Heads$.
  \item \label{sparse-distance-estimates-3}
    \begin{math}
      \min_{x \in X} \delta_j(u,x) + \delta_j(x,v) \leq \phopd{\eta}{u,v}_j
    \end{math}
    for all $u \in U$, $v \in \Heads$, and
    $\eta \in [2^{j-2},2^{j-1}]$.
  \end{mathproperties}
\end{definition}

The key difference from \cite{HJQ26} is we only have estimates to and
from a subset of vertices $X_j$, rather than estimates for every pair
(we will have $\sizeof{X_j} = \apxO{|U|/2^j}$ in the final
algorithm). Essentially, sparse distance estimates break every shortest proper
$\bigTheta{2^{j-1}}$-hop $(u, v)$ walk into two $2^{j - 1}$-hop walks,
from $u$ to some $x$ and from $x$ to $v$, for some $x \in X_j$. This
is reflected in
\cref{sparse-distance-estimates-3}. \Cref{sparse-distance-estimates-1,sparse-distance-estimates-2}
are carefully designed to work with a certain (natural) potential
function in the construction of the hop reducer, as we explain
momentarily.

\subsection{Constructing hop reducers from sparse distance estimates}

\labelsection{hop-reducer}

Suppose we are at a level $i \in \setof{i_1 + 1,\dots,L}$, and we have
sparse distance estimates at all preceding levels
$j \in \setof{i_0,\dots,i-1}$ (either given or previously
computed). The goal of this section is to use these distance estimates
to create a $\bigOmega{2^i}$-hop reducer $H_i$ for $G_i$. (The next
section will then use $H_i$ to compute the next set of sparse distance
estimates for level $i$.)

\begin{lemma}
  \labellemma{hop-reducer-construction} Let
  $i \in \setof{i_1 + 1,\dots,L}$, and suppose we have sparse distance
  estimates $(X_j,\delta_j)$ for all $j \in \setof{i_0,\dots,i-1}$. Then one can construct a
  $2^{i-2}$-hop reducer for $G_i$ with $\bigO{2^in/h}$
  vertices and $\bigO{2^im/h + \sum_{j = i_0 + 1}^{i - 1}\sizeof{U}\sizeof{X_j}}$
  edges.
\end{lemma}
\begin{proof}
  Initially, let $H_i' = G_i^+$. We call this the ``base
  subgraph''. For a vertex $v$, we let $v_i$ denote its copy in the
  base subgraph.

  For each index $j \in \setof{i_0+1,\dots,i-1}$, we augment $H_i'$
  with the following vertices and arcs that we collectively called the
  ``$j$th shortcut graph''. It was described previously in
  \refsection{bootstrap-overview} and we describe it again here. We
  first create a disjoint copy of $G_{j}^+$. For each vertex
  $v \in V_j$, let $v_j$ denote its copy in $G_j^+$. We also add a
  disjoint copy of $X_j$, which we denote $X_j'$.  For each
  $x \in X_j$, we let $x_j'$ denote its copy in $X_j'$.  For $u \in U$
  and $x \in X_j$, we add the ``shortcut'' arc $(u,x'_j)$ with length
  \begin{math}
    \len{u,x_j'}_{H_i} = \delta_{j}(u,x).
  \end{math}
  For $x \in X_j$ and $v \in \Heads$, we add the ``shortcut'' arc
  $(x_j',v_j)$ with length $\delta_{j}(x,v)$. For all
  $(y,z) \in \outcut{V_{j}}$ with $z \in V_i$, we add the ``exit'' arc
  $(y_j,z_i)$ with length $\len{y,z}$. Lastly, for each $u \in U$, we
  add the ``reset'' arc $(u_j,u_i)$ with length $0$.

  \begin{center}
    \includegraphics{figures/sparse-gadget}
  \end{center}

  For $i_0$ we create a different subgraph, because this subgraph must
  also account for all the missing distance estimates from below level
  $i_0$. We essentially create the ``layered (auxiliary) graph'' from
  \cite{Fineman24}, with some adjustments. We create $h_0 = 2^{i_0}$
  copies of $G_{i_0}^+$. For $v \in V_{i_0}$, and $\varh \in [h_0]$,
  let $v_{i_0,\varh}$ denote the copy of $v$ in the $\varh$th copy of
  $G_{i_0}^+$. For each negative arc $(u,v)$, and $\varh \in [h_0]$,
  we add the ``forward'' arc $(u_{i_0, \varh-1},v_{i_0, \varh})$ with
  length $\len{u,v}_{G}$ (letting $v_{i_0, 0} = v_i$).  For each arc
  $(x,y)$ from $V_{i_0}$ to $V_i \setminus V_{i_0}$, we add ``exit''
  arcs $(x_{i_0,\eta},y_{i})$. Lastly, for each $v \in V_{i_0}$, we
  add ``reset'' arcs $(v_{i_0,\eta},v_{i})$ of length $0$ for all
  $\eta$.

  \begin{center}
    \includegraphics{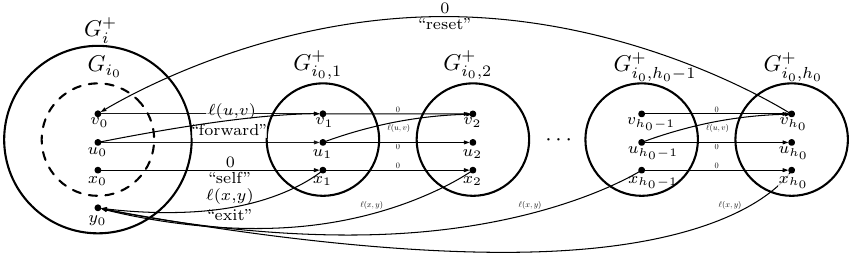}
  \end{center}

  Altogether, $H_i'$ consists of a copy of $G_i^+$ for the base graph,
  a shortcut graph for $G_j$ for each $j \in \setof{i_0+1,\dots,i-1}$,
  and a layered graph for $G_{i_0}$.  $H_i'$ has
  \begin{align*}
    \bigO{m(h_0^2 + 2^i)/h + \sum_{j = i_0 + 1}^{i - 1} |U||X_j|} %
    = \bigO{2^im/h + \sum_{j = i_0 + 1}^{i - 1} |U||X_j|}         %
  \end{align*}
  edges, since $2^i \geq 2^{2i_0} \geq h_0^2$.  The only negative arcs are the shortcut arcs in the shortcut
  graphs, the forward arcs in the layered graph, and the reset arcs in
  the shortcut and layered graphs.

  \begin{center}
    \includegraphics{figures/boosted-hop-reducer}
  \end{center}

  Every edge length in $H_i'$ either matches the length of the
  underlying edge in $G_i$, is a length-$0$ arc between two copies of
  the same vertex, or (in the case of shortcut arcs) overestimates the
  distance between the underlying endpoints in $G_i$. More explicitly,
  for any two vertices $u,v \in V_i$, and any arc in $H_i'$ from an
  auxiliary copy of $u$ to an auxiliary copy of $v$, the length of
  that arc in $H_i'$ dominates the distances from $u$ to $v$ in
  $G_i$. It follows that for all $u,v \in V_i$, the distance from
  $u_i$ to $v_i$ in $H_i'$ dominates the distance from $u$ to $v$ in
  $G_i$.

  We define potentials $\varphi$ over the vertices of $H_i'$ to
  neutralize the forward and shortcut arcs as follows.  First, in the
  base subgraph, we set $\varphi(v_{i}) = 0$ for all $v \in
  V_i$. Next, consider the $j$th shortcut graph for
  $j \in \setof{i_0 + 1,...,i-1}$. For $x \in X_j$, we set
  $\pote{x_j'} = \hopd{2^{j-1}}{V_j,x}_j$. For $v \in V_j$, we set
  \begin{math}
    \pote{v_j} = \hopd{2^{j}}{V_j,v}_j.
  \end{math}
  Lastly, in the layered graph for $G_{i_0}$, we reweight the
  vertices per the the construction in \cite{Fineman24}; namely,
  \begin{math}
    \pote{v_{i_0,\varh}} = \hopd{\varh}{V_{i_0},v}_{i_0}
  \end{math}
  for $\varh \in [h_0]$.  We claim that $\pote$ is valid and
  neutralizes all the negative edges except for the reset arcs. It
  is easy to see that $\pote$ is valid over each disjoint copy of
  $G_j^+$ (for any $j$).

  Consider an exit arc $(y_j,z_i)$ of a $j$th shortcut graph. Since
  $z$ is not in the $2^j$-hop negative reach of $U$, we have
  \begin{align*}
    \len{y_j,z_i}_{H_i',\pote} = \pote{y_j} + \len{y,z} - \pote{z_i}
    =
    \hopd{2^j}{V_j,y}_j + \len{y,z} \geq 0.
  \end{align*}
  Similarly, $\pote$ maintains the nonnegativity of each exit arc
  $(x_{i_0,\eta},y_i)$ in the layered graph because the endpoint $y$
  is not in the $2^{i_0}$-hop negative reach.

  Consider the forward arcs in the $j$th shortcut graph, for
  $j \in \setof{i_0+1,\dots,i-1}$. For the first type of shortcut arc,
  $(u_i,x_j')$ where $u \in U$ and $x \in X_j$, we have
  \begin{align*}
    \len{u_i,x'_j}_{H_i', \pote} %
    = 0 + \delta_j(u,x) - \hopd{2^{j-1}}{V_j,x}_j %
    \geq
    0
  \end{align*}
  by \cref{sparse-distance-estimates-1} of sparse distance
  estimates. For the second type of shortcut arc, $(x_j',v_j)$ where $x \in
  X_j$ and $v \in \Heads$, we have
  \begin{align*}
    \len{x_j',v_j}_{H_i',\pote} %
    =                           %
    \hopd{2^{j-1}}{V_j,x}_j +
    \delta_j(x,v) - \hopd{2^j}{V_j,x}_j %
    \geq
    0
  \end{align*}
  by \cref{sparse-distance-estimates-2} of sparse distance estimates.

  Lastly, consider the forward arcs $(u_{i_0,\varh-1},v_{i_0,\varh})$
  in the layered graph, where $\eta \in [h_0]$. As observed in
  \cite{Fineman24}, we have
  \begin{align*}
    \len{u_{i_0,\varh-1},v_{i_0,\varh}}_{H_i',\pote}
    =                           %
    \hopd{\varh-1}{V_{i_0},u}_{i_0} + \len{u,v} - \hopd{\varh}{V_{i_0},v}_{i_0}
    \geq
    0
  \end{align*}
  since
  \begin{math}
    \hopd{\varh-1}{V_{i_0},u}_{i_0} + \len{u,v}
  \end{math}
  is the length of some $\varh$-hop walk in $G_{i_0}$ from $V_{i_0}$ to $v$.

  This addresses all of the arcs of $H_i'$ except for the reset arcs in
  the shortcut and layered graphs. Letting the $H_i = H'_{i,\pote}$,
  the reset arcs are the only negative arcs in $H_i$.

  It remains to prove that $H_i$ is a $2^{i-2}$-hop reducer for $G_i$.
  Let $W : s \leadsto t$ be a proper $\varh$-hop walk in $G_i$. We
  claim there is a $\roundup{\varh/2^{i-2}}$-hop walk
  $W' : s_i \leadsto t_i$ in $H_i$ with length at most the length of
  $W$. It suffices to prove the claim for $s \in U$.  We prove the
  claim by induction on the number of hops in $W$, where $\varh=0$
  hops is immediate.

  For each index $j$, let $W_j$ be the maximal $2^{j - 1}$-hop prefix
  of $W$, and let $\bar{W}_j$ be the remaining suffix. We have two
  cases.

  \subparagraph{Case 1.} Suppose $W_j$ is contained in $G_j$ for all
  $j \in \setof{i_0, \dots, i - 1}$.  Let $j = i_0$ if
  $\varh \leq 2^{i_0 - 1}$, $j = i - 1$ if $\varh \geq 2^{i - 1}$, and
  $j = \logup[2]{\varh}$ otherwise. Let $y \in V_j$ be the last vertex
  on $W_j$. We will map $W_j$ to a $1$-hop walk
  $W_j': s_i \leadsto y_i$ in $H_i$ via the $G_j^+$-gadget, using the
  reset edge to return to the base graph. We have two subcases
  depending on whether $j = i_0$.

  \subparagraph{Case 1.a.} Suppose $j = i_0$ and all of $W_j$ is
  contained in $G_{i_0}$. We will use the layered graph to embed
  $W_j$.  Let $\varvarh$ be the number of hops in $W_j$. Let
  $W_j': s_i \leadsto y_i$ be the walk in $H_i$ that first follows the
  corresponding auxiliary edges from $s_i$ to $y_{\varvarh}$ in the
  layered graph, then takes self-arcs from $y_\varvarh$ to $y_{h_0}$
  before taking the reset arc to $y_0$.  More explicitly, the
  $\varvarvarh$th hop $(u,v)$ in $W_j$ is
  mapped to the forward arc
  $(u_{i_0,\varvarvarh-1},v_{i_0,\varvarvarh})$ between the
  $(\varvarvarh-1)$th and $\varvarvarh$th layers (or
  $(u_i,v_{i_0,\varvarvarh})$ for $\varvarvarh = 1$), and the subwalk of
  nonnegative arcs between the $(\varvarvarh-1)$th and $\varvarvarh$th
  hop are mapped to the corresponding walk in the $\varvarvarh$th
  layer $G_{i_0,\varvarvarh}^+$ (or $G_i$ for $\varvarvarh = 1$).
  $W_j'$ has 1 hop, from the last reset arc $(y_{h_0},y_0)$, and the
  same length as $W_j$.

  \begin{center}
    \includegraphics{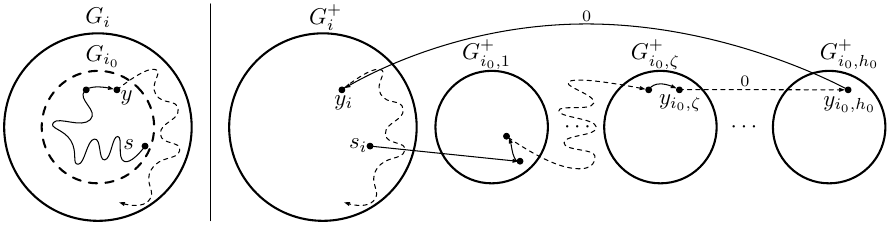}
  \end{center}

  \subparagraph{Case 1.b.} Now suppose $j > i_0$. We embed $W_j$ using the $j$th shortcut graph
  instead. Let $v \in \Heads$ be the head of the last hop in $W_{j}$
  and split $W_j$ at $v$ into $W_{j, 1}: s \leadsto v$ and
  $W_{j, 2}: v \leadsto y$. Since $W_{j,1}$ has between $2^{j - 2}$ and $2^{j - 1}$
  hops, by \cref{sparse-distance-estimates-3}, there is a vertex
  $x \in X_j$ such that
  $\delta_j(s,x) + \delta_j(x,v) \leq \len{W_{j,1}}$. The walk $W_j'$
  from $s_i$, to $x_j'$, to $v_j$, followed by the shortest walk in
  $G_j^+$ from $v_j$ to $y_j$, and ending with the reset arc to $y_i$,
  has total length
  \begin{align*}
    \len{W_j'}_{H_i} = \len{W_j'}_{H_i'} = \delta_j(s,x) +
    \delta_j(x,v) + \hopd{0}{v,y}
    \leq \len{W_{j,1}} + \len{W_{j,2}}
    \leq \len{W_j},
  \end{align*}
  and one hop.

  \begin{center}
    \includegraphics{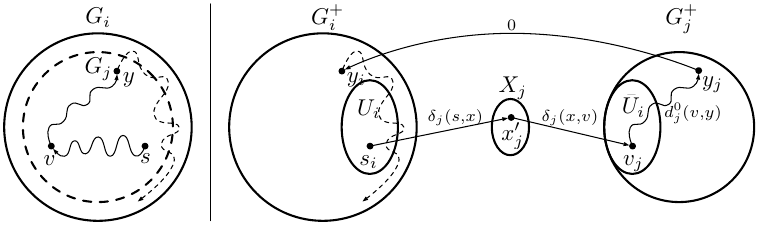}
  \end{center}

  In either case, we have a 1-hop $(s_i, y_i)$-walk $W_j'$ with
  $\len{W_j'}_{H_i} \leq \len{W_j}_{G_i}$. Meanwhile, either
  $W_j = W$; or $j = i-1$, $y \in U$, and $\bar{W}_j: y \leadsto t$ is
  a $(\varh-2^{i-2})$-hop walk in $G_i$. In the former, $W_j'$ is a
  1-hop $(s_i,t_i)$-walk in $H_i$ with length at most the length of
  $W_j = W$, as desired. In the latter, by induction on the number of
  hops, there is a $(\roundup{\varh/2^{i-2}} - 1)$-hop walk
  $\bar{W}_j'$ from $y_i$ to $t_i$ with length
  $\len{\bar{W}_j'}_{H_i} \leq \len{\bar{W}_j}_{G_i}$. Concatenated
  together, $W_j'$ and $\bar{W}_{j}'$ gives the desired
  $\roundup{\varh/2^{i-2}}$-hop $(s_i,t_i)$-walk.

  \subparagraph{Case 2.} Suppose $W_j$ is not contained in $G_j$ for
  some $j \in \setof{i_0, \dots, i - 1}$. Let $j$ be the first such
  index. Let $(y,z) \in \outcut{V_j}$ be the arc where $W_j$ first
  steps out of $V_j$. Let $q$ be the first negative vertex after
  $(y,z)$ if one exists; otherwise let $q = t$. Let
  $W_q: s \leadsto q$ be the prefix of $W$ up to $q$ and let
  $\bar{W}_q: q \leadsto t$ be the remaining suffix. Let $\varvarh$ be
  the number of hops in $W_q$. Again we have two subcases depending on
  whether $j = i_0$.

  \subparagraph{Case 2.a.} Suppose $j = i_0$. We first embed the prefix of $W_q$ up to $y$
  through the layered graph similarly to case 1: beginning from
  $s_i$, the hops in $W_j$ map to forward arcs from one layer to the
  next, and the subwalks of nonnegative arcs between arcs are retraced
  within the appropriate layer, until reaching $y_{i_0,\varvarh}$ in
  the $\varvarh$th layer. We then take the exit arc
  $(y_{i_0,\varvarh},z_i)$ and retrace the remaining suffix of
  $W_q$ through the base subgraph $G_i^+$ until ending at
  $q_i$. Letting $W_q' : s_i \leadsto q_i$ denote this walk, $W_q'$
  has $0$ hops, and the same length as $W_q$.

  \begin{center}
    \includegraphics{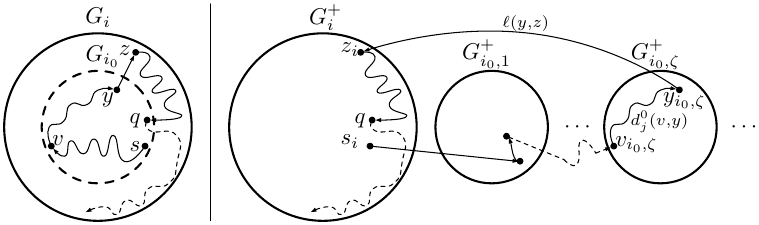}
  \end{center}

  \subparagraph{Case 2.b.} Suppose instead that $j > i_0$. Let $v$ be the head of the final hop
  in $W_q$. Let $W_{q, 1}$ be the prefix of $W_q$ from $s$ to $v$,
  $W_{q, 2}$ be the subpath from $v$ to $y$, and $W_{q, 3}$ be the
  suffix of $W_q$ from $z$ to $q$. Since $W_{j - 1}$ is contained in $G_{j - 1}$,
  $\varvarh \in [2^{j - 2}, 2^{j - 1}]$. By
  \cref{sparse-distance-estimates-3},
  $\delta_j(s,x) + \delta_j(x,v) \leq \len{W_{q_1}}$ for some
  $x \in X_j$. Let $W_{q,1}'$ be the two-edge walk from $s_i$ to $v_j$
  via $x_j'$. Let $W_{q_2}'$ be the shortest walk from $v_j$ to $y_j$
  through $G_j^+$, and let $W_{q,3}$ be the shortest walk from $z_i$
  to $q_i$ in the base subgraph $G_i^+$.
  Concatenating $W_{q, 1}'$,
  $W_{q, 2}'$, the exit arc $(y_j, z_i)$, and $W_{q, 3}'$ gives a
  0-hop $(s_i, q_i)$-walk $W_q'$ with length
  \begin{align*}
    \len{W_q'}                  %
    &=\delta_j(s,x) + \delta_j(x,v) + \hopd{0}{v,y}_j + \len{y,z} +
      \hopd{0}{z,q}
    \\
    &\leq
      \len{W_{q_1}} + \len{W_{q,2}} + \len{y,z} + \len{W_{q,3}}
      =
      \len{W_q}.
  \end{align*}
  \begin{center}
    \includegraphics{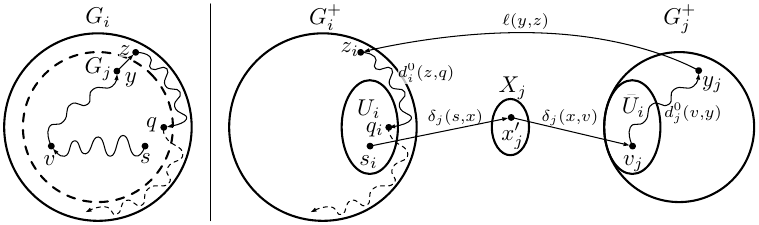}
  \end{center}

  Thus, whether or not $j = i_0$, we have a 0-hop $(s_i,q_i)$-walk $W_q'$
  with length equal to $W_q$. Meanwhile, either $\bar{W}_q$ is empty,
  or $q \in U$ and $W_q$ is a $(\varh - \varvarh)$-hop walk. Either
  trivially in the former case or by induction in the latter, there is
  a $\roundup{(\varh-\varvarh)/2^{i-2}}$-hop walk
  $\bar{W}_q': q_i \leadsto t_i$ in $H_i$ with
  $\len{\bar{W}_q'}_{H_i} \leq \len{\bar{W}_q}_{G_i}$. Together,
  $W_q'$ and $\bar{W}_q'$ give the desired
  $\roundup{\varh/2^{i-2}}$-hop walk in $H_i$.

  Between the two cases, we conclude that $H_i$ is a $2^{i-2}$-hop
  reducer for $G_i$.
\end{proof}

\subsection{Computing sparse distance estimates from hop reducers}

\labelsection{distance-estimates}

We now turn to the other half of the construction, where the hop
reducer $H_i$ is used to construct sparse distance estimates at level
$i$. We first describe the subroutine at a high-level.

The initial idea comes from \cite{HJQ26}. Recall that the sparse
distance estimates at level $i$ only need to compete with proper walks
with $\bigTheta{2^{i}}$-hops
(\cref{sparse-distance-estimates-3}). Consider the shortest such
walks. Since each such walk has at least $\bigOmega{2^j}$ negative
edges, a uniformly random sample of negative edges of size
$\bigO{\sizeof{U} \log{n} / 2^j}$ will hit all of these walks with
high probability. If we compute $\bigOmega{2^i}$-hop distances in
$H_i$ to and from all the sampled vertices, then the distances induced
by going through the sampled vertices will be at least as good as any
of these shortest walks.

The only catch is that the distances we compute in $H_i$ may be much
better than the actual proper $\bigTheta{2^i}$-hop distances in $G_i$ because $\varh$
hops in $H_i$ can capture arbitrarily more than $2^{i - 2} \varh$ hops in
$G_i$. (Too small of a distance estimate ultimately prevents us from
neutralizing the shortcut graphs constructed in
\refsection{hop-reducer} while preserving the nonnegativity of the
exit arcs.) Thus in a second step we manually detect and increase
overly optimistic distances to satisfy
\cref{sparse-distance-estimates-1,sparse-distance-estimates-2}.

\begin{lemma}
  \labellemma{reducer->distance-estimates} Given a $2^{i-2}$-hop
  reducer $H_i$ for $G_i$ with $m_i$ edges and $n_i$ vertices, one can
  compute level $i$ sparse distance estimates $(X_i,\delta_i)$ with
  $\sizeof{X_i} \leq \bigO{\sizeof{U}\log{n} / 2^i}$ with high
  probability in $\bigO{|U|\mu_i \log{n} / 2^i + 2^{2i}\mu / h}$
  time, where $\mu_i = m_i + n_i \log n_i$.
\end{lemma}

\begin{proof}
  Let $X_i \subseteq U$ sample $\bigO{\sizeof{U} \log{n} / 2^i}$
  negative vertices uniformly at random.
  Let
  \begin{align*}
    \delta_i(u,x) = \max{\hopd{2}{u, x}_{H_i}, \hopd{2^{i-1}}{V_i,x}_i}
  \end{align*}
  for $u \in U$ and $x \in X_i$ and let
  \begin{align*}
    \delta_i (x,v) = \max{\hopd{2}{x, v}_{H_i}, \hopd{2^{i}}{V_i,v}_i - \hopd{2^{i-1}}{V_i,x}_i}
  \end{align*}
  for $x \in X_i$ and $v \in \Heads$.
  Because $H_i$ is a $2^{i - 2}$-hop reducer, we have
  \begin{align*}
    d_i(u, x) \leq \hopd{2}{u,x}_{H_i} \leq \hopd{2^{i - 1}}{u,x}_i
    \text{ and }
    d_i(x, v) \leq \hopd{2}{x, v}_{H_i} \leq \hopd{2^{i - 1}}{x, v}_i.
  \end{align*}

  To verify \cref{sparse-distance-estimates-1}, let $u \in U$ and
  $x \in X_j$. We have $\delta_i(u,x) \geq d(u,x)$ because distances
  in $H_i$ always overestimate true distances in $G_i$, as desired.
  Similarly, \cref{sparse-distance-estimates-2} follows from $H_i$
  always overestimating true distances in $G_i$.

  Now consider \cref{sparse-distance-estimates-3}. Let
  $u \in U, v \in \Heads$, and let $W$ be a shortest proper $\eta$-hop
  walk from $u$ to $v$ with $2^{i - 2} \leq \eta \leq 2^{i - 1}$ hops.
  $X_i$ samples at least vertex $x$ visited by $W$ with high
  probability. For such a vertex $x$, we have
  \begin{align*}
    \delta_i(u,x) = \max{\hopd{2}{u,x}_{H_i}, \hopd{2^{i-1}}{V_i,x}_i}
    \leq
    \hopd{2^i-1}{u,x}
  \end{align*}
  because $H_i$ is a $2^{i-2}$-hop reducer.
  Next we bound
  \begin{math}
    \delta_i(x,v).
  \end{math}
  We first have
  \begin{align*}
    \hopd{2}{x,v}_{H_i} \leq \hopd{2^{i-1}}{x,v}_i
  \end{align*}
  because $H_i$ is a $2^{i-2}$-hop reducer.
  We also have
  \begin{align*}
    \hopd{2^i}{V_i,v}_i \leq \hopd{2^{i-1}}{V_i,x}_i + \hopd{2^{i-1}}{x,v}_i
  \end{align*}
  by the triangle inequality. Thus
  \begin{align*}
    \delta_i(x,v) = \max{\hopd{2}{x,v}_{H_i}, \hopd{2^{i}}{V_i,v}_i - \hopd{2^{i-1}}{V_i,x}}
    \leq
    \hopd{2^{i-1}}{u,x}.
  \end{align*}
  Altogether we have
  \begin{align*}
    \delta_i(u, x) + \delta_i(x, v) %
    \leq                            %
    \hopd{2^{i - 1}}{u, x}_i + \hopd{2^{i - 1}}{x, v}_i %
    \leq                                                %
    \len{W} = \hopd{2^{i - 1}}{u, v}_i,            %
  \end{align*}
  as desired. Taking the union bound over all $u$, $v$, and $\eta$, we
  conclude that $(X_i, \delta_i)$ are level $i$ sparse distance
  estimates satisfying
  \cref{sparse-distance-estimates-1,sparse-distance-estimates-2,sparse-distance-estimates-3}
  with high probability.

  As for the running time, computing
  $\hopd{2^{i - 1}}{{V_i, \cdot}}_i$ and
  $\hopd{2^{i}}{{V_i, \cdot}}_i$ in $G_i$ takes
  $\bigO{2^{2i} \mu / h}$ time.  Computing shortest $2$-hop paths to
  and from each $x \in X_i$ in $H_i$ takes
  $\bigO{|X_i|\mu_i} = \bigO{|U|\mu_i \log{n} / 2^i}$ time. Taking
  these distance computations and setting the estimates
  $\delta_i(\cdot, \cdot)$ takes
  $\bigO{|U||X_i|} = \bigO{|U|^2 \log{n} / 2^i}$ time. The bottleneck
  is the $\bigO{\sizeof{U} \mu_i \log{n} / 2^i}$ time computing
  $\bigO{1}$-hop distances in $H_i$ for each vertex in $X_i$.
\end{proof}

\subsection{Bootstrapping it all together}

\labelsection{bootstrap-altogether}

We conclude this section by constructing an $O(h)$-hop reducer for
$G$.  As in \cite{HJQ26}, we alternate between constructing hop
reducers from distance estimates by
\reflemma{hop-reducer-construction} and computing distance estimates
from \reflemma{reducer->distance-estimates} to build hop reducers,
with the negative reach expanding iteration by iteration. The main
difference here is we that we take as input distance estimates for
levels $i_0 + 1$ to $i_1$ to initialize the bootstrapping process to
begin from $G_{i_1}$, rather than from $G_1$.
\begin{lemma}
  \labellemma{full-reducer-construction} Let $U$ be a set of negative
  vertices that can negatively reach $\bigO{n\varh / h}$ vertices for
  all $\varh \geq h_0$. Suppose we are given as input sparse distance
  estimates $(\delta_j,X_j)$ at level $j$ for all
  $j \in \setof{i_0 + 1,\dots,i_1}$. Let
  $\Delta = \sum_{j = i_0 + 1}^{i_1} \sizeof{X_j}$.  Then one can
  compute an $(h/2)$-hop reducer for $G_U$ with
  $\bigO{m + |U|^2 \log{n} / h_0}$ edges and $\bigO{n}$ vertices with
  high probability in
  \begin{math}
    \bigO{                      %
      h\mu + |U|\mu\log{n}^2/h +
      \sizeof{U}^2 \Delta \log{n} / h_0^2 +
      \sizeof{U}^3\log{n}^2/h_0^4      %
    }                           %
  \end{math}
  randomized
  time.
\end{lemma}
\begin{proof}
  We first apply \reflemma{hop-reducer-construction} to construct a
  $(2^{i_1 - 1})$-hop reducer $H_{i_1 + 1}$ for $G_{i_1 + 1}$. We
  then follow the bootstrapping approach in \cite{HJQ26}: for all $i$
  from $i_1 + 1$ to $L - 1$, we apply
  \reflemma{reducer->distance-estimates} with $H_i$ to calculate
  sparse distance estimates $(X_i,\delta_i)$ at level $i$, and then
  apply \reflemma{hop-reducer-construction} to construct a
  $(2^{i - 1})$-hop reducer for $G_{i + 1}$. At the end of the
  process, we have the desired $(2^{L - 2})$-hop reducer for
  $G_L = G$.

  We now calculate the runtime. The time spent constructing each
  reducer, by \reflemma{hop-reducer-construction}, is
  \begin{align*}
    \bigO{\sum_{i = i_1 + 1}^L |H_i| + \frac{2^{2i}\mu}{h}} %
    =                                                        %
    \bigO{h\mu + \sum_{i = i_1 + 1}^L \sizeof{H_i}}.
  \end{align*}
  Note that the second term is negligible as it will be dominated by
  any single shortest path computation done in $H_i$. Let $m_i$ and
  $n_i$ denote the number of edges and vertices in $H_i$, and let
  $\mu_i = m_i + n_i \log n_i$. By
  \reflemma{reducer->distance-estimates}, for $j \geq i_1 + 1$, we have
  $\sizeof{X_j} \leq \bigO{|U|\log n / 2^j}$. Therefore,
  \begin{align*}
    \mu_i
    &= \bigO{\frac{2^i\mu}{h} + \sum_{j = i_0 + 1}^{i - 1}|U||X_j|}                 %
      = \bigO{\frac{2^i\mu}{h} + \sum_{j = i_0 + 1}^{i_1} \sizeof{U}\sizeof{X_j} + \sum_{j = i_1 + 1}^{i - 1} \sizeof{U}\sizeof{X_j}} \\ %
    &= \bigO{\frac{2^i\mu}{h} + \sizeof{U}\Delta + \sum_{j = i_1 + 1}^{i - 1} \frac{\sizeof{U}^2\log n}{2^i}} %
      = \bigO{\frac{2^i\mu}{h} + \sizeof{U}\Delta + \frac{\sizeof{U}^2 \log n}{h_0^2}}. %
  \end{align*}
  The time to compute the required distance estimates is
  \begin{align*}
    \bigO{\sum_{i = i_1 + 1}^{L - 1} \frac{\sizeof{U}\mu_i \log{n}}{2^i}}
    &= \bigO{\sum_{i = i_1 + 1}^{L = 1} \frac{2^i|U|\mu\log n}{2^ih} + \frac{\sizeof{U}^2\Delta\log n}{2^i} + \frac{\sizeof{U}^3 \log^2 n}{2^i h_0^2}} \\
    &= \bigO{\frac{|U|\mu\log^2 n}{h} + \frac{\sizeof{U}^2\Delta \log n}{h_0^2} + \frac{\sizeof{U}^3\log ^2 n}{h_0^4}},
  \end{align*}
  as desired.
\end{proof}

\section{Putting it all together}

\labelsection{altogether}

This section presents the final algorithm and proves
\reftheorem{best}.  Let us take stock of the components developed so
far. \reflemma{extract-sandwich} allows us to extract a remote set $U$
of negative vertices, and \reflemma{full-reducer-construction}
constructs an $O(h)$-hop reducer for $G_U$.  However,
\reflemma{full-reducer-construction} also takes as input distance
estimates at levels $i_0 +1$ through $i_1$ to initialize the
bootstrap process.  Here we bring in the ideas of
\refsection{layered-sparsification} to complete the construction. We
use layered sparsification and recursion to neutralize $G_{i_1}$, and
then use the neutralized $G_{i_1}$ to generate the missing distance
estimates.

Our algorithm is actually divided into two regimes based on the graph
density. Let $m_0 = \apxTheta{n^{\parof{33- 7 \sqrt{17}} / 4}}$. When
$m \geq m_0$, then the algorithm will basically follow the description
above. When $m \leq m_0$, we add a preprocessing step that improves
the running time slightly.

We start with the algorithm and analysis in the non-sparse regime
where $m \geq m_0$.
\begin{theorem}
  \label{dense-sssp}
  For $m \geq m_0$, single-source shortest paths with real-valued
  weights can be computed with high probability in
  $\apxO{m n^{(7-\sqrt{17})/4}}$ randomized time.
\end{theorem}

\begin{proof}
  More precisely, we prove that a graph with $m$ edges and $k$
  negative vertices can be neutralized with high probability in
  $\apxO{\mu k^{\frac{7 - \sqrt{17}}{4}} + k^{10 - 2 \sqrt{17}}}$
  randomized time (including when $m \leq m_0$). The claimed running
  time follows from taking $k = n$, running Dijkstra's algorithm in
  the neutralized graph and observing that the first term dominates
  the second when $m \geq m_0$.

  Recall that the algorithm iteratively neutralizes negative vertices
  until no negative edges remain. Consider an iteration starting with
  $k$ negative edges. Let $h$ and $h_0$ be parameters to be determined
  with $\sqrt{k} \geq h \geq h_0 \geq \bigOmega{\log n}$, and $h_0$ a
  power of 2.  Our goal in this iteration is to either neutralize
  $\bigOmega{\sqrt{kh_0}}$ negative edges or return a negative cycle.

  We first apply \reflemma{extract-sandwich}, and either (a) directly
  neutralize $\bigOmega{\sqrt{kh_0}}$ (and end the iteration), (b) find
  a negative cycle (and return), or (c) obtain a set $U$ of
  $\bigOmega{\sqrt{kh_0}}$ negative vertices and a valid potential
  $\pote_1$, such that for all $\eta \geq h_0$, $U$ can $\eta$-hop
  negatively reach at most $n \eta / h$ vertices in the reweighted
  graph $G_{\pote_1}$. Suppose the iteration continues in case (c). The
  goal for this iteration shifts to neutralizing $U$. For ease of
  notation, let $G$ denote $G_{U,\pote_1}$ for the rest of the
  iteration.

  Let $i_0 = \log[2]{h_0}$ and $i_1 = 2 i_0$. We will neutralize $U$ by
  building a $\bigO{h}$-hop reducer via
  \reflemma{full-reducer-construction}, and then compute Johnson's
  potentials neutralizing $U$ in this hop reducer. To this end we must
  first compute distance estimates at level $i$ for all
  $i \in \setof{i_0 + 1,\dots,i_1}$.

  For $i \in \setof{i_0,\dots,i_1}$, let $G_i = (V_i,E_i)$ be the subgraph
  induced by the $2^{i}$-hop negative reach of $U$.  $G_{i_0}$ has
  $\bigO{m h_0/n}$ edges and $\bigO{n h_0/h}$ vertices, and $G_{i_1}$ has
  $\bigO{m h_0^2 / h}$ edges and $\bigO{n h_0^2 / h}$ vertices.

  We neutralize $G_{i_1}$ with layered sparsification and recursion.
  We apply \reflemma{layered-sparsification} to $G_{i_1}$, using the
  $h_0$-hop negative reach $G_{i_0}$ as layers, which returns a graph
  $H = (V_H,E_H)$, potentials $\poteA: V_H \to \reals$, and mappings
  $\pi_1,\pi_2: V_{i_1} \to V_H$ as described in
  \reflemma{layered-sparsification}. We then recurse on $H_{\poteA}$
  to compute potentials $\poteB$ neutralizing $H_{\poteA}$ (with high
  probability), and then use $H_{\poteA,\poteB}$ to compute Johnson
  potentials $\poteC : V_{i_1} \to \reals$ neutralizing $G_{i_1}$.

  The potentials $\poteC$ also neutralize $G_{i}$ for all
  $i \leq i_1$. For each $i \in \setof{i_0 + 1,...,i_1}$, we use the
  neutralized graph $G_{i,\poteC}$ to produce distance estimates at
  level $i$ by essentially the same techniques as in
  \refsection{distance-estimates}. There are multiple ways to
  go about this. One way is to repeat the steps of the proof of
  \reflemma{reducer->distance-estimates}, using $G_{i,\poteC}$ instead
  of a hop reducer to compute the distance estimates.

  Alternatively, to apply \reflemma{reducer->distance-estimates}
  directly, we can create a hop reducer $H_i$ for $G_i$ where we
  initially start with disjoint copies of $G_i^+$ and
  $G_{i,\poteC}$. For each vertex $v$, let $v'$ denotes its copy in
  $G_i^+$ and $v''$ its copy in $G_{i,\poteC}$. For each vertex $v$, we
  add an arc $(v',v'')$ of weight $-\poteC{v}$ and an arc $(v'',v')$
  of weight $\poteC{v}$. By translating $\poteC$ to be nonnegative, we
  can assume that the edges from $G_{i,\poteC}$ to $G_i^+$ are
  nonnegative, leaving only the edges from $G_i^+$ to $G_{i,\poteC}$
  as possibly negative.

  It is easy to see that $H_i$ is an $n$-hop reducer. A $0$-hop
  $(s,t)$-walk $W$ maps directly to the same $0$-hop $(s',t')$-walk in
  $G_i^+$.  An $(s,t)$-walk $W$ with at least one hop in $G_i$ can be
  embedded as a one-hop $(s',t')$-walk in $H_i$ of the same length by
  taking one hop from $s'$ to $s''$, retracing $W$ with $0$ hops in
  $G_{i,\poteC}$, and taking the nonnegative edge back from $t''$ to
  $t'$. Now, applying \reflemma{reducer->distance-estimates} to $H_i$
  gives us the desired distance estimates $(X_i,\delta_i)$ at level
  $i$.

  We submit the distance estimates
  $\setof{(X_i,\delta_i)\where i = i_0 + 1,\dots,i_1}$ to
  \reflemma{full-reducer-construction}, returning a $\bigO{h}$-hop
  reducer $H$. We use $H$ to compute Johnson potentials for $G$,
  neutralizing $U$.

  This describes a single iteration of the algorithm. Each randomized
  subroutine succeeds with high probability, hence each iteration also
  succeeds with high probability. There are at most
  $\apxO{\sqrt{k /h_0}}$ iterations, so all the iterations also
  succeed with high probability. Altogether the algorithm successfully
  either neutralizes the graph, or returns a negative cycle, with high
  probability.

  \subparagraph{Running time analysis.} It remains to bound the
  running time.  Let $\Time{m,n,k}$ denote the running time on a graph
  with $m$ edges, $n$ vertices, and $k$ negative edges.  Consider a
  single iteration. The first step, applying
  \reflemma{extract-sandwich} to extract the remote set $U$, takes
  \begin{align*}
    \bigO{h\mu \log^2 n + h^2\sqrt{k/h_0^3}\log^2 n}
  \end{align*}
  time. Next, sparsifying $G_{i_0}$ with
  \reflemma{layered-sparsification} and recursing on the sparsified
  graph $H_i$, takes
  \begin{align*}
    \Time{\bigO{m_i}, \bigO{n_i}, \bigO{\sizeof{U} \log{n}/h_0}}
    =
    \bigO{\Time{m h_0^2 / h, n h_0^2 / h, \sqrt{k/h_0} \log{n}}}
  \end{align*}
  time.  Let $\mu_i = m_i + n_i \log n_i$ for
  $i \in \setof{i_0 + 1,\dots,i_1}$. It takes $\bigO{\mu_{i_1}}$ time to use
  $H_{i_1}$ to compute the Johnson potentials $\poteC$ neutralizing
  $G_{i_1}$.  Next we compute distance estimates for each $i$ with
  \reflemma{reducer->distance-estimates}. For each $i$, this takes
  \begin{align*}
    \bigO{\frac{\sizeof{U} \mu_i \log{n}}{2^i} + \frac{2^{2i} \mu}{h}}
    =
    \bigO{\frac{\sizeof{U} \mu \log{n}}{h} + \frac{2^{2i} \mu}{h}}
  \end{align*}
  time for each $i$.  Thus we spend
  \begin{align*}
    \bigO{\sum_{i = i_0 + 1}^{i_1}\frac{\sizeof{U} \mu \log{n}}{h} + \frac{2^{2i} \mu}{h}}
    =
    \bigO{\frac{\sizeof{U} \mu \log{n}^2}{h} + h\mu}
  \end{align*}
  computing distance estimates for levels
  $i \in \setof{i_0 + 1,\dots,i_1}$.

  The next step uses the distance estimates to construct a
  $\bigOmega{h}$-hop reducer with
  \reflemma{full-reducer-construction}. By
  \reflemma{full-reducer-construction}, noting that
  \begin{math}
    \Delta = \sum_{i=i_0 + 1}^{i_1} \sizeof{X_i} = \bigO{\sizeof{U} \log{n} / h_0},
  \end{math}
  it takes
  \begin{align*}
    \bigO{h \mu + \frac{\sizeof{U} \mu \log^2 n}{h} +
    \frac{\sizeof{U}^3 \log{n}^2}{h_0^3}}
  \end{align*}
  time to construct the $\bigOmega{h}$ hop reducer $H$.

  Lastly we calculate the Johnson potentials with a $\parof{\sizeof{U}
    / h}$-hop distance computation in $H$. Since $H$ has $\bigO{m +
    \sizeof{U}^2 \log{n} / h_0}$ edges and $\bigO{n}$, this takes
  \begin{align*}
    \bigO{\frac{\sizeof{U}^3 \log{n}}{h h_0} + \frac{\sizeof{U}\mu}{h}}
  \end{align*}
  time. Putting everything together, a single iteration with
  $\sizeof{U} = \bigO{\sqrt{kh}}$ remote vertices takes
  \begin{align*}
    &\apxO{h m + \frac{\sizeof{U} m}{h} +
      \frac{\sizeof{U}^3}{h_0^3} +
      \Time{\frac{m h_0^2}{h},\frac{n h_0^2}{h},
      \frac{\sizeof{U}}{h_0}}}
    \\
    &=
      \apxO{h m + \frac{\sqrt{kh_0} m}{h} +
      \frac{k^{1.5}}{h_0^{1.5}} +
      \Time{\frac{m h_0^2}{h},\frac{n h_0^2}{h},
      \sqrt{\frac{k}{h_0}}}}
  \end{align*}
  time.  Altogether, $\Time{m,n,k}$ satisfies the following recursion.
  \begin{align*}
    \Time{m,n,k}
    =
    \apxO{
    \sqrt{\frac{k}{h_0}}
    \parof{h m + \frac{\sqrt{kh_0} m}{h} +
    \frac{k^{1.5}}{h_0^{1.5}} +
    \Time{\frac{m h_0^2}{h},\frac{n h_0^2}{h},
    \sqrt{\frac{k}{h_0}}}
    }
    }.
  \end{align*}
  For $h = \bigTheta{k^{\frac{\sqrt{17}-3}{4}}}$ and
  $h_0 = \bigTheta{ k^{\sqrt{17}-4}}$, this recursion is satisfied by
  \begin{align*}
    \Time{m,n,k} = \apxO{m k^{\frac{7 - \sqrt{17}}{4}} + k^{10 - 2 \sqrt{17}}},
  \end{align*}
  as desired.
\end{proof}

In the very sparse regime where $m \leq m_0$, the algorithm above
would give an $\apxO{n^{10 - 2 \sqrt{17}}}$ randomized running time,
where $10 - 2 \sqrt{17} \leq 1.75379$. The following algorithm, which
adds a preprocessing step to the previous algorithm, improves on this
bound for the very sparse regime.
\begin{theorem}
  \label{sparse-sssp}
  For $m \leq m_0$, single-source shortest paths with real-valued edge
  lengths can be computed with high probability in
  $\apxO{(m n)^{\parof{66-2\sqrt{17}}/67}}$ randomized time.
\end{theorem}

\begin{proof}
  Let $k$ denote the number of negative vertices in $G$. The proof of
  \cref{dense-sssp} actually gives an algorithm that runs in
  $\apxO{mk^{\frac{7 - \sqrt{17}}{4}} + k^{10 - 2\sqrt{17}}}$ time,
  for all $m$. When $m < k^{\frac{33 - 7\sqrt{17}}{5}}$, the second
  term becomes the bottleneck. In this regime, a slight increase in
  $m$ in exchange for a slight decrease in $k$ lowers the overall
  running time. We will construct a new graph $H$ preserving distances
  in $G$, that decreases the number of negative edges $k$ at the cost
  of increasing the total number of edges $m$. The ratio of edges to
  negative edges in $H$ will be exactly the threshold delineating the
  sparse regime.

  We construct such an $H$ with layered sparsification. For a
  parameter $h$ to be determined, we apply layered sparsification to
  $G$ with $h$ layers of $G$ itself (\reflemma{layered-sparsification}
  with $r = 1$) to get $H, \varphi, \pi_0$, and $\pi_1$, as described
  in \reflemma{layered-sparsification}. $H$ has $m_H = \apxO{hm}$
  edges and $k_H = \apxO{k/h}$ negative vertices. We then neutralize
  $H$ with Theorem \ref{dense-sssp} and subsequently neutralize $G$ by
  computing Johnson's potentials via the neutralized $H$.

  We choose
  \begin{math}
    h = \parof{\frac{k^{33 - 7\sqrt{17}}}{m^4}}^{\frac{1}{37 -
        7\sqrt{17}}}
  \end{math},
  decreasing $k$ and increasing $m$ just enough to return to the
  non-sparse regime of \cref{dense-sssp}. This gives a final runtime
  of
  \begin{align*}
    \apxO{m_H k_H^{\frac{7 - \sqrt{17}}{4}}} = \apxO{(mk)^{\frac{66 - 2\sqrt{17}}{67}}},
  \end{align*}
  as desired.
\end{proof}

\paragraph{Conclusion.}
Together, the running times from \cref{dense-sssp} for $m \geq m_0$
and \cref{sparse-sssp} for $m \leq m_0$ give
\reftheorem{sssp}.

We conclude with a few remarks reflecting on the improvement in
running time from \cite{HJQ26} to \reftheorem{best}. Both algorithms
ultimately use the bootstrapping construction and both algorithms are
bottlenecked by the $\apxO{m \sizeof{U}/h}$ time computing distance
estimates at each level. The subtle difference is that here we start
the bootstrapping from the $h_0$-hop negative reach, for some $h_0$
polynomial in $k$, instead of at the smallest, $\bigO{\log n}$-hop
negative reach. In turn, we do not need $U$ to be
$\bigO{\eta,\eta/h}$-remote for $\eta < h_0$, and we can increase the
size of the remote set returned by \reflemma{extract-sandwich}
substantially from $\sqrt{k \log n}$ to $\sqrt{k h_0}$. The layered
sparsification technique introduced in
\refsection{layered-sparsification} gives us a faster recursive
algorithm to neutralize the $h_0$-hop negative reach, allowing us to
extend the value of $h_0$ without this step becoming a bottleneck.

In summary, the algorithm here accelerates \cite{HJQ26} by
neutralizing more negative edges by a polynomial,
$\sqrt{h_0 / \log n}$ factor in each iteration. Starting the bootstrap
construction at a moderate $h_0$-hop negative reach allows the
sandwich size to increase by a $\sqrt{h_0}$ factor. The efficiency
gained by layered sparsification allows us to increase $h_0$. These
are the two driving factors in the faster running time.

At the moment, $\apxO{m \sizeof{U} / h}$ time spent computing distance
estimates feels like a natural cost in the bootstrapping approach, and
it remains an interesting challenge to overcome this bottleneck.

We also remark that the algorithm here goes beyond \cite{HJQ26} for
very sparse graphs. For sparse graphs, the $\apxO{m^{4/5} n}$ from
\cite{HJQ26} is obtained from the $\apxO{m n^{4/5}}$ algorithm for
sufficiently dense graphs by rebalancing internal parameters, but the
algorithmic steps are all the same. We go beyond simple
reparametrization in \cref{sparse-sssp} by using layered
sparsification --- increasing the total edges while decreasing the
negative edges --- before applying the algorithm from the non-sparse
regime.

\begin{remark}
  We did not try to optimize logarithmic factors in favor of a simpler
  exposition. If one chooses the parameters $m_0$, $h$, and $h_0$ to
  also account for logarithmic factors, then letting
  \begin{math}
    m_0 = n^{\parof{33-\sqrt{17}}/4} / \log^7 n,
  \end{math}
  one gets
  \begin{align*}
    \bigO{\mu n^{\frac{7-\sqrt{17}}{4}} \log^5 n}
  \end{align*}
  randomized time for $m \geq m_0$, and
  \begin{align*}
    \bigO{\parof{\mu n}^{\frac{66 - 2 \sqrt{17}}{67}} \log^{\frac{394 -
    16\sqrt{17}}{67}} n}
  \end{align*}
  randomized time for $m \leq m_0$.
\end{remark}

\printbibliography

@InProceedings{BNW22,
  author       = {Aaron Bernstein and Danupon Nanongkai and Christian
                  Wulff{-}Nilsen},
  title        = {Negative-Weight Single-Source Shortest Paths in
                  Near-linear Time},
  year         = 2022,
  booktitle    = {63rd {IEEE} Annual Symposium on Foundations of
                  Computer Science, {FOCS} 2022, Denver, CO, USA,
                  October 31 - November 3, 2022},
  pages        = {600-611},
  doi          = {10.1109/FOCS54457.2022.00063},
  url          = {https://doi.org/10.1109/FOCS54457.2022.00063},
  crossref     = {DBLP:conf/focs/2022},
  bibsource    = {dblp computer science bibliography,
                  https://dblp.org}
}

@proceedings{DBLP:conf/focs/2022,
  title        = {63rd {IEEE} Annual Symposium on Foundations of Computer Science, {FOCS}
                  2022, Denver, CO, USA, October 31 - November 3, 2022},
  publisher    = {{IEEE}},
  year         = {2022},
  url          = {https://doi.org/10.1109/FOCS54457.2022},
  doi          = {10.1109/FOCS54457.2022},
  isbn         = {978-1-6654-5519-0},
  bibsource    = {dblp computer science bibliography, https://dblp.org}
}

@proceedings{DBLP:conf/stoc/2024,
  editor       = {Bojan Mohar and
                  Igor Shinkar and
                  Ryan O'Donnell},
  title        = {Proceedings of the 56th Annual {ACM} Symposium on Theory of Computing,
                  {STOC} 2024, Vancouver, BC, Canada, June 24-28, 2024},
  publisher    = {{ACM}},
  year         = {2024},
  url          = {https://doi.org/10.1145/3618260},
  doi          = {10.1145/3618260},
  bibsource    = {dblp computer science bibliography, https://dblp.org}
}

@Article{DI17,
  author       = {Yefim Dinitz and Rotem Itzhak},
  title        = {Hybrid Bellman-Ford-Dijkstra algorithm},
  journal      = {J. Discrete Algorithms},
  year         = 2017,
  volume       = 42,
  pages        = {35-44},
  doi          = {10.1016/J.JDA.2017.01.001},
  url          = {https://doi.org/10.1016/j.jda.2017.01.001},
  bibsource    = {dblp computer science bibliography,
                  https://dblp.org}
}

@InProceedings{Fineman24,
  author       = {Jeremy T. Fineman},
  title        = {Single-Source Shortest Paths with Negative Real
                  Weights in $\tilde{O}(m n^{8/9})$ Time},
  year         = 2024,
  booktitle    = {Proceedings of the 56th Annual {ACM} Symposium on
                  Theory of Computing, {STOC} 2024, Vancouver, BC,
                  Canada, June 24-28, 2024},
  pages        = {3-14},
  doi          = {10.1145/3618260.3649614},
  url          = {https://doi.org/10.1145/3618260.3649614},
  crossref     = {DBLP:conf/stoc/2024},
  bibsource    = {dblp computer science bibliography,
                  https://dblp.org}
}

@Article{Goldberg95,
  author       = {Andrew V. Goldberg},
  title        = {Scaling Algorithms for the Shortest Paths Problem},
  journal      = {{SIAM} J. Comput.},
  year         = 1995,
  volume       = 24,
  number       = 3,
  pages        = {494-504},
  doi          = {10.1137/S0097539792231179},
  url          = {https://doi.org/10.1137/S0097539792231179},
  bibsource    = {dblp computer science bibliography,
                  https://dblp.org}
}

@Article{Dijkstra59,
  author       = {Edsger W. Dijkstra},
  title        = {A note on two problems in connexion with graphs},
  journal      = {Numerische Mathematik},
  year         = 1959,
  volume       = 1,
  pages        = {269-271},
  doi          = {10.1007/BF01386390},
  url          = {https://doi.org/10.1007/BF01386390},
  bibsource    = {dblp computer science bibliography,
                  https://dblp.org}
}

@article{Bellman58,
 ISSN = {0033569X, 15524485},
 URL = {http://www.jstor.org/stable/43634538},
 author = {Richard Bellman},
 journal = {Quarterly of Applied Mathematics},
 number = {1},
 pages = {87--90},
 publisher = {Brown University},
 title = {On A Routing Problem},
 urldate = {2024-06-30},
 volume = {16},
 year = {1958}
}

@book{Ford56,
author="Ford, Lester R.",
title="Network Flow Theory.",
address="Santa Monica, CA",
year = 1956,
publisher="RAND Corporation"
}

@InProceedings{Shimbel55,
  author =       {Alfonso Shimbel},
  title =        {Structure in communication nets},
  booktitle = {Proceedings of the Symposium on Information Networks},
  year =      1955,
  pages =     {199--203},
  address =   {New York, New York},
  publisher = {Polytechnic Press of the Polytechnic Institute of Brooklyn}}

@Article{FT87,
  author       = {Michael L. Fredman and Robert Endre Tarjan},
  title        = {Fibonacci heaps and their uses in improved network
                  optimization algorithms},
  journal      = {J. {ACM}},
  year         = 1987,
  volume       = 34,
  number       = 3,
  pages        = {596-615},
  doi          = {10.1145/28869.28874},
  url          = {https://doi.org/10.1145/28869.28874},
  bibsource    = {dblp computer science bibliography,
                  https://dblp.org}
}

@InProceedings{Moore59,
  author =       {Edward F. Moore},
  title =        {The shortest path through a maze},
  booktitle = {Proceedings of an International Symposium on the Theory of Switching 1957, Part II},
  year =      1959,
  address =   {Cambridge, Massachusetts},
  publisher = {Harvard University Press}}

@Article{Johnson77,
  author       = {Donald B. Johnson},
  title        = {Efficient Algorithms for Shortest Paths in Sparse
                  Networks},
  journal      = {J. {ACM}},
  year         = 1977,
  volume       = 24,
  number       = 1,
  pages        = {1-13},
  doi          = {10.1145/321992.321993},
  url          = {https://doi.org/10.1145/321992.321993},
  bibsource    = {dblp computer science bibliography,
                  https://dblp.org}
}

@InProceedings{HJQ25a,
  ids = {HJQ25},
  author       = {Yufan Huang and Peter Jin and Kent Quanrud},
  title        = {Faster single-source shortest paths with negative
                  real weights via proper hop distance},
  year         = 2025,
  booktitle    = {Proceedings of the 2025 Annual {ACM-SIAM} Symposium
                  on Discrete Algorithms, {SODA} 2025, New Orleans,
                  LA, USA, January 12-15, 2025},
  pages        = {5239-5244},
  doi          = {10.1137/1.9781611978322.178},
  url          = {https://doi.org/10.1137/1.9781611978322.178},
  crossref     = {DBLP:conf/soda/2025},
  bibsource    = {dblp computer science bibliography,
                  https://dblp.org}
}

@proceedings{DBLP:conf/soda/2025,
  editor       = {Yossi Azar and
                  Debmalya Panigrahi},
  title        = {Proceedings of the 2025 Annual {ACM-SIAM} Symposium on Discrete Algorithms,
                  {SODA} 2025, New Orleans, LA, USA, January 12-15, 2025},
  publisher    = {{SIAM}},
  year         = {2025},
  url          = {https://doi.org/10.1137/1.9781611978322},
  doi          = {10.1137/1.9781611978322},
  isbn         = {978-1-61197-832-2},
  bibsource    = {dblp computer science bibliography, https://dblp.org}
}

@InProceedings{DMMSY25,
  author       = {Ran Duan and Jiayi Mao and Xiao Mao and Xinkai Shu
                  and Longhui Yin},
  title        = {Breaking the Sorting Barrier for Directed
                  Single-Source Shortest Paths},
  year         = 2025,
  booktitle    = {Proceedings of the 57th Annual {ACM} Symposium on
                  Theory of Computing, {STOC} 2025, Prague, Czechia,
                  June 23-27, 2025},
  pages        = {36-44},
  doi          = {10.1145/3717823.3718179},
  url          = {https://doi.org/10.1145/3717823.3718179},
  crossref     = {DBLP:conf/stoc/2025},
  bibsource    = {dblp computer science bibliography,
                  https://dblp.org}
}

@proceedings{DBLP:conf/stoc/2025,
  editor       = {Michal Kouck{\'{y}} and
                  Nikhil Bansal},
  title        = {Proceedings of the 57th Annual {ACM} Symposium on Theory of Computing,
                  {STOC} 2025, Prague, Czechia, June 23-27, 2025},
  publisher    = {{ACM}},
  year         = {2025},
  url          = {https://doi.org/10.1145/3717823},
  doi          = {10.1145/3717823},
  isbn         = {979-8-4007-1510-5},
  bibsource    = {dblp computer science bibliography, https://dblp.org}
}

@Article{CKLPPS25,
  author       = {Li Chen and Rasmus Kyng and Yang P. Liu and Richard
                  Peng and Maximilian Probst Gutenberg and Sushant
                  Sachdeva},
  title        = {Maximum Flow and Minimum-Cost Flow in Almost-Linear
                  Time},
  journal      = {J. {ACM}},
  year         = 2025,
  volume       = 72,
  number       = 3,
  pages        = {19:1--19:103},
  doi          = {10.1145/3728631},
  url          = {https://doi.org/10.1145/3728631},
  bibsource    = {dblp computer science bibliography,
                  https://dblp.org},
  addendum = {Preliminary version in FOCS, 2022}
}

@misc{LLRZ25,
      title={Shortcutting for Negative-Weight Shortest Path}, 
      author={George Z. Li and Jason Li and Satish Rao and Junkai Zhang},
      year={2025},
      eprint={2511.12714},
      archivePrefix={arXiv},
      primaryClass={cs.DS},
      url={https://arxiv.org/abs/2511.12714}, 
}

@misc{Version1,
      title={From Hop Reduction to Sparsification for Negative Length Shortest Paths}, 
      author={Kent Quanrud and Navid Tajkhorshid},
      year=2025,
      eprint={2511.18253v1},
      archivePrefix={arXiv},
      primaryClass={cs.DS},
      url={https://arxiv.org/abs/2511.18253v1},
      month = {November}
}

\appendix
\crefalias{section}{appendix}
\section{Faster betweenness reduction and updated running times}

\label{recursive-betweenness-reduction}

This section integrates the improved betweenness reduction subroutine
of \cite{LLRZ25}, calculates the implied running times for the two
algorithms described above, and describes a third algorithm that
better leverages the new subroutine.

\subsection{Betweenness reduction}

``Betweenness reduction'' is a preprocessing step introduced by
\citet{Fineman24}. For three vertices $s,t,v \in V$ and
$h \in \nnintegers$, we say that $v$ is \defterm{(weakly) $h$-hop
  negatively between} $s$ and $t$ if
\begin{align*}
  \hopd{h}{s, v} + \hopd{h}{v, t} < 0.
\end{align*}
For parameters $b,h \in \naturalnumbers$, \cite{Fineman24} gave a
subroutine computing a valid potential $\poteA$ such that with high
probability, for any two vertices $s,t$, at most $n/h$ vertices are
negatively between $s$ and $t$ in the reweighted graph
$G_{\poteA}$. \cite{Fineman24} algorithm ran in $\apxO{mbh + b^2 n}$
randomized time and the $b^2 n$ term was essentially removed in follow
up work \cite{HJQ26}. Betweenness reduction is a critical step for
extracting a remote set.  This article, which is primarily focused on
how to neutralize a remote set of edges, and not on how to obtain a
remote set of edges, hides the betweenness reduction step in
\reflemma{extract-sandwich}. The running time in
\reflemma{extract-sandwich} directly reflects the running time of the
betweenness reduction step, and betweenness reduction is a bottleneck
for all the algorithms in this article.

\citet{LLRZ25} obtained the following improved bounds for betweenness
reduction.
\begin{lemma}[{\cite[Lemma 11]{LLRZ25}}]
  \labellemma{recursive-betweenness-reduction} Let
  $b,h \in \naturalnumbers$.  One can compute a valid potential
  $\poteA$ so that any two vertices have at most $n/b$ vertices
  $h$-hop negatively between them in $G_{\poteA}$, in running time
  bounded by $\bigO{\parof{m + n \log n}h}$ plus the time to
  neutralize a real-weighted graph with $\bigO{m h}$ edges and
  $\bigO{b \log n}$ negative edges.
\end{lemma}
\begin{proof}[Proof sketch]
  We give a sketch of the proof that highlights the main ideas, as it
  is relatively simple. We refer to \cite{LLRZ25} for
  full details and analysis. The new subroutine can be interpreted as
  simulating the approach in \cite{Fineman24} with two key ideas that
  improve the efficiency.

  We first sketch the $\apxO{m bh + nb^2}$-time betweenness reduction
  algorithm from \cite{Fineman24}.  Let $X$ be a set of $\apxO{b}$
  vertices uniformly at random. We construct an auxiliary graph $H$,
  beginning with $G^+$.  For all $x \in X$ and $v \in V$, we add
  ``shortcut'' arcs $(x,v)$ and $(v,x)$ of length $d^h(x,v)$ and
  $d^h(v,x)$, respectively. (These edge lengths can be computed via an
  $h$-hop distance computation to and from $x$ in $G$, for each
  $x \in X$, in $\apxO{mbh}$ time.) The only negative arcs in $H$ are
  the shortcut arcs.

  \begin{center}
    \includegraphics{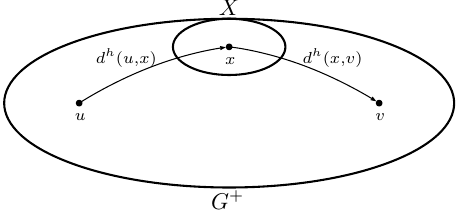}
  \end{center}

  Suppose we compute Johnson's potentials $\varphi$ neutralizing
  $H$. (Only $\bigO{\sizeof{X}}$ hops are needed to compute these
  potentials, because every negative edge is incident to a vertex in
  $X$.) Since $H$ contains $G^+$, the restriction of $\varphi$ to the
  vertices in $G^+$ is also a valid potential for $G$. Moreover,
  because $\varphi$ neutralizes the shortcut arcs in $H$, for all $s,t
  \in V$ and $x \in X$, we have
  \begin{align*}
    \hopd{h}{s,x}_{G,\pote} + \hopd{h}{x,t}_{G,\pote} = \len{s,x}_{H,\pote} +
    \len{x,t}_{H,\pote} \geq 0.
  \end{align*}
  As argued in \cite{Fineman24}, the inequality above implies that all
  but a $(1/b)$-fraction of the vertices are not negatively between
  $s$ and $t$ with high probability. The overall running time was
  $\apxO{mbh + b^2 n}$, where the $b^2n$ term accounts for the $bn$
  shortcut edges added to $H$. The primary bottleneck is the
  $\apxO{mbh}$ time computing the lengths of the shortcut edges.

  \cite{LLRZ25} accelerates Fineman's subroutine with two key
  ideas. First, the new subroutine uses $h$-layer auxiliary graphs to
  effectively simulate $h$-hop distance computations without having to
  compute the distances explicitly to and from every sampled
  vertex. This comes at a cost of increasing the total number of edges
  from $m$ to $mh$.  Second, the auxiliary graph is arranged so that,
  after adding potentials that neutralize the negative ``forward'' arcs
  between the layers, the remaining auxiliary graph only has
  $\sizeof{X}$ negative arcs. The key point is that the relatively few
  remaining negative arcs can be neutralized recursively. Thus the
  overall running time is that of neutralizing $\sizeof{X} = \apxO{b}$
  negative arcs in a graph of $mh$ edges.

  The construction is essentially as follows. We construct a graph that
  initially consists of $2h+1$ copies of $G^+$. We call these copies
  $G_{-h}^+, G_{1-h}^+,\dots,G_{-1}^+,G_0^+,G_1^+,...,G_{h-1}^+,G_h^+$.
  For each vertex $v$ and index $i$, let $v_i$ denote its copy in
  $G_i^+$.  For every negative arc $(u,v)$ and index $i < h$, we add a
  ``forward'' arc $(u_i,v_{i+1})$ of the same length. For each vertex
  $v$ and index $i < h$, we add a length-0 ``self'' arc
  $(v_i,v_{i+1})$.  Lastly, for every sampled negative vertex $x \in X$,
  we add an auxiliary ``reset'' arc $(x_h,x_{-h})$ of length $0$.

  \begin{center}
    \includegraphics{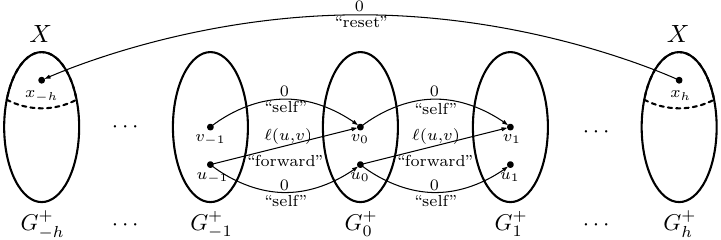}
  \end{center}

  Relative to the construction in \cite{Fineman24}, we can interpret the
  ``top'' layers $G_1^+,...,G_h^+$ as simulating the shortcut arcs
  $d^h(u,x)$ for $u \in V$ and $x \in X$, and the ``bottom'' layers
  $G_{-h}^+,...,G_{-1}^+$ as simulating the shortcut arcs $d^h(x,v)$ for
  $x \in X$ and $v \in V$. The ``reset'' arcs then allow us to
  concatenate shortcut arcs going through $X$.

  Consider the potentials $\pote$ over $V$ defined by
  \begin{align*}
    \pote{v_i} =
    \begin{cases}
      \hopd{i}{V,v_i} &\text{if } i \geq 0 \\
      -\hopd{-i}{v_i,V} &\text{if } i \leq 0.
    \end{cases}
  \end{align*}
  It is easy to see that $\pote$ is valid over $H$ and neutralizes all
  of the forward arcs. The only negative arcs in $H_{\varphi}$ are the
  reset arcs $(x_h,x_{-h})$ for $x \in S$.

  One recursively neutralizes $H_{\poteA}$ to produce another
  potential $\poteB$.  The combined potentials $\poteA+\poteB$,
  restricted to $G_0^+$, give valid potentials for $G$. For all
  $s,t \in V$ and $x \in X$, we have
  \begin{align*}
    \hopd{h}{s,x}_{G,\poteA+\poteB} + \hopd{h}{x,t}_{G,\poteA+\poteB}
    =
    \hopd{0}{s_0,x_h}_{H,\poteA + \poteB} + \len{x_h,x_{-h}}_{H,\poteA +
    \poteB} + \hopd{0}{x_{-h},t_0}_{H,\poteA+\poteB} \geq 0.
  \end{align*}
  It follows from \cite{Fineman24} that with high probability, all
  pairs of vertices have at most $n/b$ vertices $h$-hop negatively between
  them in $G_{\poteA+\poteB}$.
\end{proof}

With \reflemma{recursive-betweenness-reduction}, the running time from
\reflemma{extract-sandwich} improves to the following. If we only need
betweenness reduction for a single hop parameter (as in the recursive
sparsification algorithm), then the following bounds follow directly
from \reflemma{recursive-betweenness-reduction}, Lemma 4.1 in
\cite{HJQ25}, and Lemma 3.4 from \cite{HJQ26}.

\begin{lemma}
  \label{extract-sandwich-2}
  Let $h \geq \bigOmega{\log n}$. One can compute, with high
  probability and in the time to neutralize $\bigO{b \log n}$ negative
  edges in a graph of $\bigO{mh}$ edges, either:
  \begin{compactmathresults}
  \item A negative cycle.
  \item Valid potentials $\varphi$ neutralizing a set of
    $\bigOmega{\sqrt{\numN h}}$ negative vertices.
  \item Valid potentials $\varphi$ and a set $U$ of negative vertices
    such that:
    \begin{compactmathproperties}
    \item $\sizeof{U} \geq \bigOmega{\sqrt{\numN h}}$.
    \item $U$ has $h$-hop negative reach of at most $n/b$ vertices.
    \end{compactmathproperties}
  \end{compactmathresults}
\end{lemma}

To obtain remoteness over a range of hop parameters, as in the
bootstrapping algorithm, one can extend
\reflemma{recursive-betweenness-reduction} and obtain the
following. (We refer the reader to \cite[Lemma 29]{LLRZ25} for the
extension of \reflemma{recursive-betweenness-reduction}.)

\begin{lemma}
  \labellemma{extract-sandwich-3}
  Let $h_0 = \bigOmega{\log n}$ and $h \geq h_0$. One can compute,
  with high probability, and in $\apxO{mh}$ time plus the time to
  neutralize $\apxO{b}$ negative edges in a graph of $\apxO{mh_0}$ edges,
  either:
  \begin{compactmathresults}
  \item A negative cycle.
  \item Valid potentials $\varphi$ neutralizing a set of
    $\bigOmega{\sqrt{\numN h_0}}$ negative vertices.
  \item Valid potentials $\varphi$ and a set $U$ of negative vertices
    such that:
    \begin{compactmathproperties}
    \item $\sizeof{U} \geq \bigOmega{\sqrt{\numN h_0}}$.
    \item $U$ has $h_0$-hop negative reach of size $n / b$.
    \item $U$ has $\varh$-hop negative reach of size $n\varh/h$ for
      all $\varh \geq h_0$.
    \end{compactmathproperties}
  \end{compactmathresults}
\end{lemma}

\subsection{Plugging into the recursive algorithm}

We now update the running time of the simple recursive algorithm based
on the improved running time of \cref{extract-sandwich-2}.

Recall the recursive algorithm from
\refsection{layered-sparsification}. For a parameter $h$, the
algorithm repeatedly extracts an $h$-remote set $U$ of size
$\sqrt{k h}$, creates an $h$-layered graph of $\bigO{m}$ edges,
sparsifies the number of negative edges to
$\bigO{\sqrt{k/h} \log{n}}$, recursively neutralizes the sparsified
layered graph, and uses the neutralized graph to compute Johnson's
potentials for $G_U$.  Letting $\Time{m,k}$ denote the running time to
neutralize $k$ negative edges in a graph of size $m$,
\cref{extract-sandwich-2} improves the running time of the first step
to $\Time{mh,h}$. Overall we have a recurrence of the form
\begin{align*}
  \Time{m,k} = \apxO{\sqrt{k/ h}\parof{\Time{mh,h} + \Time{m,\sqrt{k /
  h}}}}.
\end{align*}
This recurrence minimized by $h = k^{\frac{\sqrt{2} - 1}{\sqrt{2} + 1}}$, giving
\begin{align*}
  \Time{m,k} = \apxO{m k^{1/\sqrt{2}}}.
\end{align*}

\subsection{Updating the bootstrapping running time}

One can similarly revise the running time of the bootstrapping
algorithm that obtained the bounds in \reftheorem{sssp}.  Here
balancing parameters is not as straightforward. Some (natural)
parameterizations lead to the same $\apxO{mn^{1/\sqrt{2}}}$ running
time (for sufficiently dense graphs) as above. One can do slightly
better by reparameterizing the bootstrapping algorithm as follows. Let
$\alpha \approx 0.7044$ and consider an iteration with $k$ negative
vertices.  We set $h = k^{1 - \alpha}$ and
$h_0 = k^{2/\alpha - 2 - \alpha}$, and use \reflemma{extract-sandwich-3} to
extract a set of negative vertices $U$ with
$\sqrt{kh^{1 - \alpha}h_0^\alpha}$ negative vertices that is
$(\varh, h/\varh)$-remote for all $\varh \in [h_0, h]$. The rest of
the algorithm is identical.  All together, this achieves a running
time of $\apxO{mk^{\alpha}} \leq \bigO{m n^{.7044}}$ when $m$ is above
some density threshold. We omit additional details because this
algorithm appears to be slower than the following algorithm, which is
also simpler. We also believe that a tweak of the bootstrapping
algorithm can bring the running time very slightly, but would still be
still slower than the ensuing algorithm.

\subsection{The twice- (or thrice-) recursive algorithm}

We now describe a third approach, that lies somewhere between the
recursive sparsification algorithm and the bootstrapping algorithm in
complexity.  It uses the layered sparsification technique from the
simple recursive sparsification algorithm, and the shortcut edges from
the bootstrapping algorithm, but does not iteratively bootstrap hop
reducers from a small negative reach to the entire graph as in the
actual boostrapping algorithm.  We originally called it the
``twice-recursive algorithm'' because it recurses on two graphs to
neutralize a remote set, although perhaps it ought to be called the
``thrice-recursive algorithm'' now that \cref{extract-sandwich-2} adds
a third recursive call via the betweenness reduction step.

Let $\alpha \approx 0.69562077$ be the unique real root of the
polynomial $p(x) = x^3 + 2 x^2 + x - 2$. Let
$\gamma = 1+ \frac{\alpha^2\parof{1-\alpha}}{2(2+\alpha)} \approx
1.0273194$. Let $m_0 = n^{\gamma}$.

\begin{lemma}
  \labellemma{twice-recursive-dense}
  For $m \geq m_0$, single-source shortest paths with
  real-weighted edge weights
  can be solved with high probability in $\apxO{m n^{\alpha}}$
  randomized time.
\end{lemma}
\begin{proof}
  More precisely, we prove that a graph with $m$ edges and $k$
  negative edges, and $m \geq k^{\gamma}$, can be neutralized with
  high probability in $\apxO{m k^{\alpha}}$ randomized time. The
  overall claim follows then follows from substituting $k = n$ and
  running Dijkstra's algorithm in the neutralized graph.

  As with all other algorithms we've discussed, we neutralize the negative arcs
  iteratively. Let $k$ be the number of negative arcs/vertices at the
  beginning of an iteration. Let
  $h = \apxTheta{k^{\alpha^2 / (2 + \alpha)}}$ and
  $b = \apxTheta{k^{\parof{1-\alpha}}}$. ($h$ and $b$ are balanced in
  hindsight.)  We first invoke \cref{extract-sandwich-2} to either
  directly neutralize $\apxTheta{\sqrt{kh}}$ vertices (and repeat),
  identify a negative cycle (and exit), or obtain an $(h,b)$-remote
  set $U$ of size $\sizeof{U} = \apxTheta{\sqrt{kh}}$.

  We continue in the third case. For ease of notation let $G = G_U$
  for the rest of the iteration.  Similar to the recursive
  sparsification algorithm, we construct a weighted graph
  $H = (V_H, E_H)$, potentials $\pote: V_H \to \reals$, and maps
  $\mapstart, \mapend: V \to V_H$ such that with high probability, (a)
  $H_{\pote}$ has $\apxO{m}$ edges and $\apxO{\sqrt{k/h}}$ negative
  edges, and (b)
  $\distance{\mapstart{s}, \mapend{t}}_H = \distance{s,t}$ for all
  $s,t \in V$.  We neutralize $H_{\pote}$ recursively, use the
  neutralized graph to compute Johnson potentials neutralizing $G$,
  and conclude the iteration.

  Next we describe and analyze the construction of the auxiliary graph
  $H$. The main point of $H$ is that it effectively preserves
  distances in $G$ while sparsifying the number of negative edges. One
  can interpret the following construction as extending the layered
  graph in \refsection{layered-sparsification} with ``shortcut arcs''
  from the ``shortcut gadgets'' in \refsection{sparse-bootstrap}.

  Let $V_h$ be the $h$-hop negative reach of $U$, and let $G_h$ be the
  subgraph induced by $V_h$. Let $X$ sample
  $\bigO{\sizeof{U}\log{n}/b}$ negative vertices uniformly at random.
  To construct $H$, we start with two disjoint copies of $G^+$,
  denoted $G_0$ and $G_b$; $b-1$ copies of $G_h^+$, denoted
  $G_1,G_2,\dots,G_{b-1}$; and a disjoint copy of $X$, denoted
  $X'$. For a vertex $v$ and index $i$, we let $v_i$ denote the copy of
  $v$ in $G_i$ (where applicable). For a sampled vertex $x \in X$, we
  let $x'$ denote the auxiliary vertex in $X'$.

  At a high level, $H$ is a layered graph from $G_0$ to $G_b$,
  with additional shortcuts from $G_0$ to $G_b$ via $X$, and
  sparsified reset arcs from $G_b$ back to $G_0$. Specifically, for
  each negative arc $(u,v)$ and index $i \in [b]$, we have a
  ``forward'' arc $(u_{i-1},v_i)$ of the same length. For each vertex
  $v \in V_h$ and index $i \in [b]$, we have length-$0$ ``self'' arcs
  $(v_{i-1},v_i)$. For each arc $(x,y) \in \outcut{V_h}$ leaving the
  $h$-hop negative reach, we add an ``exit'' arc $(x_b, y_0)$ of
  length $\len{x,y}$. Next we add the ``shortcut'' arcs. For $u \in U$
  and $x \in X$, we add an arc $(u,x)$ of length
  $\max{d_{G_h}(u,x), d^{h/2}(V,x)}$. For $x \in X$ and
  $v \in \Heads$, we add an arc $(x,v)$ of length
  $\max{d_{G_h}(x,v), d^{h}(V,v) - d^{h/2}(V,x)}$.  Lastly, let
  $U_0 \subseteq U$ sample $\bigO{\sizeof{U} \log{n} / h}$ negative
  vertices uniformly at random. For each $u \in U_0$, we add a
  length-$0$ ``reset'' arc $(u_b,u_0)$.

  \begin{center}
    \includegraphics{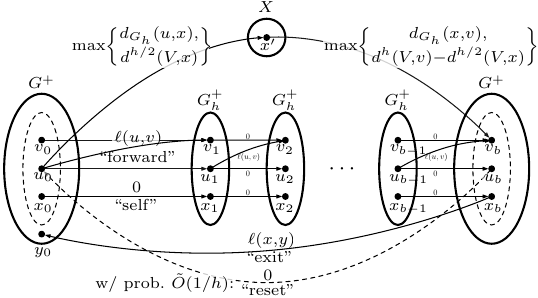}
  \end{center}

  $H$ has $\bigO{m + kh\log{n}/b} \leq \bigO{m}$ edges, noting that
  the second term is at most $m_0$ for our choices of $h$, $b$, and
  $m_0$.  Every arc $(u_i,v_j)$ in $H$ is either a copy of an arc
  $(u,v)$ in $G$, a self-arc (with $u = v$), or a shortcut
  arc. Obviously, a shortcut arc $(u_0,x')$ or $(x',v_b)$ has length
  at least $\distance{u,x}_{G_h} \geq \distance{u,x}$ and
  $\distance{x,v}_{G_h} \geq \distance{x,v}$.  Meanwhile, exactly as
  argued in \cref{lemma:reducer->distance-estimates}, we have
  \begin{math}
    \len{u_0,x'}_{H} \leq \hopd{h/2}{u,x}_{G_h}
  \end{math}
  and
  \begin{math}
    \len{x',v_b}_{H} \leq \hopd{h/2}{x,v}_{G_h}
  \end{math}
  for all $u \in U$, $v \in \Heads$, and $x \in X$.

  Define potentials $\pote$ over $H$ by
  \begin{math}
    \pote{v_0} = 0
  \end{math}
  for $v \in V$, $\pote{v_i} = \hopd{i}{V,v}$ for $v \in V_h$ and
  $i \in [b-1]$, $\pote{v_b} = \hopd{h}{V,v}$ for $v \in V$, and
  $\pote{x'} = \hopd{h/2}{V,x'}$. It
  is straightforward to verify that $\pote$ is valid over $H$ and
  neutralizes all the negative arcs except the $\sizeof{U_0}$ reset
  arcs. As in \reflemma{hop-reducer-construction}, the shortcut arcs
  are neutralized because $\len{u_0,x'}_{H} \geq \hopd{h/2}{V,x}$ and
  $\len{x',v_b}_{H} \geq \hopd{h}{V,v} - \hopd{h/2}{V,x}$ for
  $u \in U$, $v \in \Heads$, and $x \in X$.

  We first analyze the distance-preserving correctness of $H$, and
  discuss the running time to construct $H$ after.  We claim that with
  high probability, $\distance{s_0,t_b}_{H} = \distance{s,t}_{G}$ for
  all $s, t \in V$. We have
  $\distance{s_0,t_b}_{H} \geq \distance{s,t}$ for all $s,t \in V$
  because every arc from an auxiliary copy of a vertex $u$ to an
  auxiliary copy a vertex $v$ has length at least $d(u,v)$.

  The reverse inequality follows from essentially the same reasons as
  \cref{lemma:hop-reducer-construction,lemma:layered-sparsification}.
  First, observe that $H$ contains the layered sparsification graph
  from \cref{lemma:layered-sparsification}, with $b-1$ intermediate
  layers of the $h$-hop negative reach, as a subgraph. The same
  argument as in the proof of \cref{lemma:layered-sparsification}
  proves that
  \begin{align*}
    \distance{s_0,t_b}_{H} \leq \hopd{b}{s,t}
  \end{align*}
  for all $s,t \in V$.

  Next we extend the inequality above to
  \begin{math}
    d_{H}(s_0,t_b) \leq \hopd{h/2}{s,t}
  \end{math}
  for all $s,t \in V$.  To this end, we prove that
  \begin{math}
    d_{H}(s_0,t_b) \leq \hopd{h/2}{s,t}
  \end{math}
  for all $s, t \in V$ and $\eta \leq h/2$, by induction on $\eta$.

  Let $W: s \leadsto t$ be a proper $\eta$-hop walk of length
  $\hopd{\eta}{s,t}$ for $b < \eta \leq h/2$.  Here we have two cases
  depending on whether or not $W$ stays in $G_h$ between its first and
  last hop.

  In the first case, suppose $W$ stays in $G_h$ between the first and
  last hop. Here the argument is the same as case 2.b of the proof of
  \cref{lemma:hop-reducer-construction}, observing that $X$ hits at
  least one hop in $W$ with high probability, using the shortcut arcs
  to route the subwalk between the first and last hops of $W$.

  In the second case, suppose $W$ leaves $G_h$ between its first and
  last hop. Here the argument is the same as the argument from case 2
  of the proof of \cref{lemma:layered-sparsification}, splitting $W$
  at the arc $(x,y)$ leaving $G_h$, embedding the prefix and suffix by
  induction on $\eta$, and concatenating them with the exit arc
  $(x_b,y_0)$.

  Thus
  \begin{math}
    \distance{s_0,t_b}_{H} \leq \hopd{h/2}{s,t}
  \end{math}
  for all $s,t \in V$. The inequality extends to all hop lengths by
  the exact same argument as in the proof of
  \reflemma{layered-sparsification}: with high probability, $U_0$ hits
  any $(s,t)$-shortest walk once every $h/2$ hops with high
  probability. We embed each $h/2$-hop segment through $H$, and
  connect the embeddings with length-$0$ reset arcs.

  It remains to analyze the running time of the algorithm.  Consider
  the construction of $H$. It takes $\bigO{h \mu}$ to compute the
  distances $\hopd{i}{V,v}$, $i \leq h$, used to compute $\pote$ and
  to label the shortcut arcs. The other nontrivial step is computing
  the distances $\distance{u,x}_{G_h}$ and $\distance{x,v}_{G_h}$ for
  $x \in X$, $u \in U$, and $v \in \Heads$. To this end, we
  recursively neutralize $G_h$, and then compute shortest paths to and
  from each $x \in X$ in the neutralized graph. This takes one
  recursive call to a graph with $m/b$ edges and $\sizeof{U}$ negative
  edges, and $\sizeof{X}$ nonnegative shortest path computations in a
  graph of $m/b$ edges.

  Let $\Time{m,k}$ denote the running time to neutralize a graph with
  $m$ edges and $k$ negative edges. $\Time{m,k}$ is modeled
  recursively by
  \begin{align*}
    \Time{m,k} =                %
    \apxO{                      %
    \sqrt{\frac{k}{h}}\parof{   %
    \Time{mh, b} + \Time{m/b, \sqrt{kh}} + %
    \Time{m, \sqrt{\frac{k}{h}}} + \frac{m \sqrt{kh}}{b^2}}
    },                           %
  \end{align*}
  where we note that the recursive calls will also be to subproblems
  in the dense regime.  The recursion is solved by
  $\Time{m,k} \leq \apxO{m k^{\alpha}}$.
\end{proof}

For very sparse graphs with $m < n^{1+\gamma}$, similar to
\cref{sparse-sssp} for the bootstrapping algorithm, we used layered
sparsification to reduce to the dense case and then apply
\reflemma{twice-recursive-dense}.
\begin{lemma}
  \labellemma{twice-recursive-sparse} For $m \leq n^{\gamma}$,
  single-source shortest paths with real-valued edge lengths can be
  computed with
  $\apxO{\parof{mn}^{\frac{\alpha + \gamma}{1 + \gamma}}}$ randomized
  time, where $\frac{\alpha+\gamma}{1 + \gamma} \approx 0.849861$.
\end{lemma}
\begin{proof}
  Suppose there are $k \leq n$ negative edges; as with
  \reflemma{twice-recursive-dense}, we directly analyze the time to
  neutralize a graph with $m$ edges and $k$ negative edges.  Let
  $h = \apxO{(k^{\gamma}/m)^{1/1+\gamma}}$. As described in the proof
  of \cref{sparse-sssp}, let $H$ be the graph obtained by layering $G$
  $h$ times and sparsifying the negative edges by a factor of
  $\bigO{h / \log n}$.  $H$ has
  $m_H \defeq \bigO{mh} = \apxO{\parof{mk}^{\gamma/ 1 + \gamma}}$
  edges and $k_H = \apxO{k/h} = \bigO{\parof{mk}^{1/(1+\gamma)}}$
  negative edges.  In particular, $m_H \geq k_H^{\gamma}$. By
  \reflemma{twice-recursive-dense}, we neutralize $H$ with high
  probability in
  $\apxO{m_H k_H^{\alpha}} = \apxO{\parof{m n}^{\parof{\alpha +
        \gamma}/\parof{1 + \gamma}}}$ randomized time, and use the
  neutralized $H$ to compute distances for $G$.
\end{proof}

Together, \reflemma{twice-recursive-dense} and
\reflemma{twice-recursive-sparse} give the running times in
\reftheorem{sssp-updated}.

\begin{remark}
  We were aware of the ``twice recursive algorithm'' when preparing
  the initial version of this paper, but it was almost completely
  superseded by the bootstrapping algorithm.  With the previous
  running time of \reflemma{extract-sandwich} for extracting a remote
  set, it only eked out a better running time for $m < n^{1.0235}$, by
  a factor of at most $n^{0.000107}$ --- a comically small margin of
  improvement.  We felt at the time that it wasn't worth the added
  complexity to the exposition.  Evidently, the approach becomes more
  compelling under the new bounds of \cref{extract-sandwich-2}. We
  found it interesting to see the twice-recursive approach become
  relatively stronger, while the bootstrapping approach apparently
  loses its edge, with the new betweenness-reduction
  subroutine.
\end{remark}

\end{document}
